%% file: main.tex
\documentclass[11pt,a4paper,3p,preprint]{elsarticle}

\usepackage{amsthm}
\usepackage{thmtools}
\usepackage{colonequals}
\usepackage{macros}
\usepackage{url}
\usepackage{graphicx,wrapfig}
\usepackage{todonotes}
\usepackage{color}
\usepackage{booktabs}
\usepackage{xfrac,ifxetex}
\usepackage{microtype}
\usepackage{etoolbox}
\usepackage{tcolorbox,afterpage,array,caption,subcaption}
\usepackage[colorlinks=true,
  pdfauthor={Aile Ge-Ernst, Christoph Scholl, Juraj S\'\i\v c, Ralf Wimmer},
  pdftitle={Solving Dependency Quantified Boolean Formulas Using Quantifier Localization}]{hyperref}
\usepackage[ruled,linesnumbered,lined,boxed]{algorithm2e}
\usepackage{float}
\newtoggle{long}
\toggletrue{long}
\usepackage{thmtools}
\usepackage{thm-restate}

\begin{document}
\begin{frontmatter}
  \renewcommand{\thefootnote}{\fnsymbol{footnote}}
  \title{Solving Dependency Quantified Boolean Formulas \\ Using Quantifier Localization\textsuperscript{\footnotemark[1]}}
  \author{Aile Ge-Ernst$^1$}
  \author{Christoph Scholl$^1$}
  \author{Juraj S\'\i\v c$^{2,3}$}
  \author{Ralf Wimmer$^{4,1}$}
  \address{$^1$Department of Computer Science, Albert-Ludwigs-Universit\"at Freiburg \\
    Freiburg im Breisgau, Germany \\
    {\upshape\texttt{\{\href{mailto:geernsta@informatik.uni-freiburg.de}{geernsta},
                       \href{mailto:scholl@informatik.uni-freiburg.de}{scholl},
                       \href{mailto:wimmer@informatik.uni-freiburg.de}{wimmer}\}@informatik.uni-freiburg.de}}\\[\baselineskip]
    $^2$Faculty of Informatics, Masaryk University \\
    Brno, Czech Republic \\[\baselineskip]
    $^3$Faculty of Information Technology, Brno University of Technology \\
    Brno, Czech Republic \\
    {\upshape\texttt{\href{mailto:sicjuraj@fit.vutbr.cz}{sicjuraj@fit.vutbr.cz}}}\\[\baselineskip]
    $^4$Concept Engineering GmbH, Freiburg im Breisgau, Germany \\
    {\upshape\url{https://www.concept.de}}
  }
  \date{\today}

  \input{00-abstract}
\end{frontmatter}

\renewcommand{\thefootnote}{\fnsymbol{footnote}}
\footnotetext[1]{This work is an extended version of \cite{geernst-et-al-fmcad-2019}.
We added detailed proofs for all theorems, new theory on equisatisfiability under substituting subformulas (together
 with the corresponding proofs), new algorithms adjusted to the extended theory, and updated experimental results.}
\renewcommand{\thefootnote}{\arabic{footnote}}

\input{01-introduction}
\input{02-preliminaries}

\input{03-nonprenex}
\input{04-algorithm}
\input{05-experiments}
\input{06-conclusion}
\input{06a-acknowledgements}
\input{07-appendix}

\bibliographystyle{elsarticle-num}
\bibliography{abbrev,literature}
\end{document}

%% file: 00-abstract.tex
\begin{abstract}
    Dependency quantified Boolean formulas (DQBFs) are a powerful formalism, which
    subsumes quantified Boolean formulas (QBFs) and allows an explicit specification
    of dependencies of existential variables on universal variables.
    Driven by the needs of various applications which
    can be encoded by DQBFs in a natural, compact, and elegant way,
    research on DQBF solving has emerged in the past few years.
    However, research focused on closed DQBFs in prenex form (where all quantifiers
    are placed in front of a propositional formula),
    while non-prenex DQBFs have almost not been studied in the literature.
    In this paper, we provide a formal definition for syntax and semantics of
    non-closed non-prenex DQBFs and prove useful properties enabling quantifier localization.
    Moreover, we make use of our theory by integrating quantifier localization into
    a state-of-the-art DQBF solver. Experiments with prenex DQBF benchmarks, including all instances from
    the QBFEVAL'18--'20 competitions,
    clearly show that quantifier localization pays off in this context.
\end{abstract}

%% file: 01-introduction.tex
\section{Introduction}
\label{sec:introduction}
\noindent During the last two decades enormous progress in the solution of
quantifier-free Boolean formulas (SAT) has been observed. Nowadays, SAT solving
is successfully used in many applications, \eg in
planning~\cite{RintanenHN06}, automatic test pattern
generation~\cite{EggersglussD12,CzutroPLERB10}, and formal verification of hard-
and software systems~\cite{BiereCCSZ03,ClarkeBRZ01,IYG+:2008}. Motivated by the
success of SAT solvers, efforts have been made,
\eg \cite{LonsingB10,Janota2012,JanotaM15,RabeT15}, to consider the
more general formalism of quantified Boolean formulas (QBFs).

Although QBFs are capable of encoding decision problems in the PSPACE complexity
class, they are not powerful enough to succinctly encode many natural and
practical problems that involve decisions under partial information. For
example, the analysis of games with incomplete information~\cite{PetersonRA01},
topologically constrained synthesis of logic circuits~\cite{BalabanovCJ14},
synthesis of safe controllers~\cite{BloemKS14}, synthesis of fragments of
linear-time temporal logic (LTL)~\cite{CHOP:2013}, and verification of partial
designs~\cite{SchollB01,gitina-et-al-iccd-2013} fall into this category and
require an even more general formalism, which is known as \emph{dependency quantified
Boolean formulas (DQBFs)} \cite{PetersonRA01}.

Unlike QBFs, where an existential variable implicitly depends on all the
universal variables preceding its quantification level, DQBFs admit that arbitrary
dependency sets are explicitly specified. Essentially, these quantifications with
explicit dependency sets correspond to Henkin quantifiers~\cite{Henkin61}.
The semantics of a DQBF can be interpreted from a game-theoretic viewpoint as a
game played by one universal player and multiple non-cooperative existential
players with incomplete information, each partially observing the moves of the
universal player as specified by his/her own dependency set. A DQBF is true if
and only if the existential players have winning strategies. This specificity of
dependencies allows DQBF encodings to be exponentially more compact than their
equivalent QBF counterparts. In contrast to the PSPACE-completeness of QBF, the
decision problem of DQBF is NEXPTIME-complete~\cite{PetersonRA01}.

Driven by the needs of the applications mentioned above, research on DQBF
solving has emerged in the past few years, leading to solvers such as
\textsc{iDQ}~\cite{FrohlichKBV14},
HQS~\cite{gitina-et-al-date-2015,wimmer-et-al-sat-2015,WimmerKBS017},
dCAQE~\cite{TentrupR19}, iProver~\cite{Korovin08}, and DQBDD~\cite{Sic:2020}.

As an example for a DQBF, consider the formula
\begin{equation*}
  \forall x_1 \forall x_2 \exists y_1(x_1) \exists y_2(x_2) \, : \, (x_1 \wedge x_2) \equiv (y_1 \equiv y_2)
\end{equation*}
from \cite{Rabe17}.
Here $\forall x_1 \forall x_2 \exists y_1(x_1) \exists y_2(x_2)$ is called the quantifier prefix
and $(x_1 \wedge x_2) \equiv (y_1 \equiv y_2)$ the matrix of the DQBF.
This DQBF asks whether there are choices for $y_1$ only depending on the
value of $x_1$, denoted $\exists y_1(x_1)$, and for $y_2$ only depending on $x_2$,
denoted $\exists y_2(x_2)$, such that the Boolean formula after the quantifier prefix
evaluates to true for all assignments to $x_1$ and $x_2$.\footnote{We can
interpret this as a \emph{game} played by $y_1$ and $y_2$ against $x_1$ and
$x_2$, where $y_1$ and $y_2$ only have incomplete information on actions of
$x_1$, $x_2$, respectively.} The Boolean formula in turn states that the
existential variables $y_1$ and $y_2$ have to be equal iff $x_1$ and $x_2$ are
true. Since $y_1$ can only `see' $x_1$ and $y_2$ only $x_2$, $y_1$ and $y_2$
`cannot coordinate' to satisfy the constraint. Thus, the formula is false.
Now consider  a straightforward  modification of this DQBF into a QBF with only
implicit dependency sets. Changing the quantifier prefix into a
QBF quantifier prefix $\forall x_1 \exists y_1 \forall x_2 \exists y_2$ means
that $y_1$ may depend on $x_1$, but $y_2$ may depend on $x_1$ and $x_2$. In that
case the formula would be true. Changing the prefix into $\forall x_2 \exists
y_2 \forall x_1 \exists y_1$ has a similar effect.

So far, syntax and semantics of DQBFs have been defined only for closed prenex forms (see for
instance \cite{BalabanovCJ14}),
\ie for DQBFs where all quantifiers are placed in front of the matrix and all variables occurring in the
matrix are either universally or existentially quantified. In this paper, we consider quantifier
localization for DQBF, which transforms prenex DQBFs into non-prenex DQBFs for more efficient
DQBF solving.

Quantifier localization for \emph{QBF} has been used with great success for image and pre-image
computations in the context of sequential equivalence checking and symbolic model checking
where it has been called ``early quantification''.
Here existential quantifiers were moved over \emph{AND} operations \cite{GB:94,HKB:96,CCL+:97,MKR+:2000}.
In \cite{Benedetti05c} the authors consider quantifier localization for
QBFs where the matrix is restricted to conjunctive normal form (CNF).
They move universal and existential quantifiers over \emph{AND} operations
and propose a method to construct a tree-shaped quantifier structure from a
QBF instance with linear quantifier prefix. Moreover, they show how to benefit from this
structure in the QBF solving phase.
This work has been used and generalized in \cite{PigorschS09} for a QBF solver based on symbolic
quantifier elimination.

To the best of our knowledge, quantifier localization has not been considered for DQBF so far,
apart from the seminal theoretical work on DQBF by Balabanov et al. \cite{BalabanovCJ14},
which considers --~as a side remark~-- quantifier localization for DQBF, transforming prenex DQBFs into
non-prenex DQBFs. For quantifier localization they gave two propositions.
However, a formal definition of the semantics of non-prenex DQBFs was missing in that work
and, in addition, the two propositions are not sound, as we will show in our paper.

In this paper, we provide a formal definition of syntax and semantics of non-prenex non-closed DQBFs.
The semantics is based on Skolem functions and is a natural generalization of the semantics
for closed prenex DQBFs known from the literature. We introduce an alternative constructive
definition of the semantics and show that both semantics are equivalent.
Then we define rules for transforming DQBFs into equivalent or equisatisfiable DQBFs, which
enable the translation of prenex DQBFs into non-prenex DQBFs. The rules are similar to their QBF counterparts,
but it turns out that some of them need additional conditions for being sound for DQBF as well.
Moreover, the proof techniques are completely different from those for their corresponding QBF counterparts.
We provide proofs for all the rules.
Finally, we show a method that transforms a prenex DQBF into a non-prenex DQBF based on those rules.
It is inspired by the method constructing a tree-shaped quantifier structure from \cite{Benedetti05c}
and works for DQBFs with an arbitrary formula (circuit) structure for the matrix.
The approach tries to push quantifiers ``as deep into the formula'' as possible.
Whenever a sub-formula fulfills conditions, which we will specify in Section~\ref{sec:ncnp}, it is processed by symbolic quantifier elimination.
When traversing the structure back, quantifiers which could not be eliminated are pulled back into the
direction of the root. At the end, a prenex DQBF solver is used for the simplified formula.
Experimental results demonstrate the benefits of our method when applied to
a set of more than 5000 DQBF benchmarks (including all QBFEVAL'18--'20 competition~\cite{qbfeval18,qbfeval19,qbfeval20}
benchmarks).

The paper is structured as follows:
In Section~\ref{sec:preliminaries} we provide preliminaries needed to understand the paper,
including existing transformation rules for QBFs.
Section~\ref{sec:ncnp} contains the main conceptual results of the paper whereas Section~\ref{sec:algorithm}
shows how to make use of them algorithmically.
Section~\ref{sec:experiments} presents experimental results and Section~\ref{sec:conclusion}
concludes the paper.

%% file: 02-preliminaries.tex
\section{Preliminaries}
\label{sec:preliminaries}

\noindent Let $\varphi,\kappa$ be quantifier-free Boolean formulas over the set $V$ of variables
and $v\in V$. We denote by $\varphi[\sfrac{\kappa}{v}]$ the Boolean formula which results
from $\varphi$ by replacing all occurrences of $v$ (simultaneously) by $\kappa$.
For a set $V'\subseteq V$ we denote by $\assign(V')$ the
set of \emph{Boolean assignments} for $V'$, \ie
$\assign(V')=\bigl\{\mu\,\big|\,\mu:V'\to\{0,1\}\bigr\}$.
As usual, for a Boolean assignment $\mu\in\assign(V')$ and $V'' \subseteq V'$
we denote the restriction of $\mu$ to $V''$ by $\mu|_{V''}$.
For each formula $\varphi$ over $V$, a variable assignment
$\mu\in\assign(V)$ induces a truth value $0$ or $1$ of $\varphi$, which we
call $\mu(\varphi)$.
If $\mu(\varphi)=1$ for all $\mu\in\assign(V)$,
then $\varphi$ is a \emph{tautology}.
In this case we write $\vDash\varphi$.

A \emph{Boolean function} with the set of input variables $V$ is a mapping $f:\assign(V)\to\{0,1\}$.
The set of Boolean functions over $V$ is denoted by $\boolf{V}$.
The \emph{support} $\support(f)$ of a function $f\in\boolf{V}$ is defined by
$\support(f) \colonequals \{v \in V \; | \; \exists \mu, \mu' \in \assign(V) \mbox{ with } \mu(w) = \mu'(w) \, \text{ for all }w \in V\setminus \{v\}
\mbox{ and } f(\mu) \neq f(\mu')\}$.
$\support(f)$ is the set of variables from $V$ on which $f$ ``really depends''.
The constant zero and constant one function are $\fzero$ and $\fone$, respectively.
$\mathrm{ITE}$ denotes the if-then-else operator, \ie $\mathrm{ITE}(f,g,h) = (f\land g)\lor (\neg f\land h)$.

A function $f:\assign(V)\to\{0,1\}$ is \emph{monotonically increasing (decreasing)} in $v\in V$,
if $f(\mu) \leq f(\mu')$ ($f(\mu') \leq f(\mu)$) for all assignments $\mu, \mu' \in\assign(V)$
with $\mu(w) = \mu'(w)$ for all $w \in V \setminus \{v\}$ and $\mu(v) \leq \mu'(v)$.

A quantifier-free Boolean formula $\varphi$ over $V$ defines a Boolean function
$f_{\varphi}:\assign(V)\to\{0,1\}$ by $f_{\varphi}(\mu) \colonequals \mu(\varphi)$.
When clear from the context, we do not differentiate between
quantifier-free Boolean formulas and the corresponding Boolean functions,
\eg if $\varphi$ is a Boolean formula representing $f_{\varphi}$, we write
$\varphi[\sfrac{v'}{v}]$ for the Boolean function where the input variable
$v$ is replaced by a (new) input variable $v'$.

Now we consider Boolean formulas with quantifiers.
The usual definition for a \emph{closed prenex} DQBF is given as follows:

\begin{definition}[Closed prenex DQBF]
  \label{def:dqbf_cp}
  Let $V=\{x_1,\ldots,x_n,\allowbreak y_1,\ldots,y_m\}$ be a
  set of Boo\-lean variables.
  A \emph{dependency quantified Boolean formula} (DQBF) $\psi$ over $V$ has
  the form
  \begin{equation*}
    \psi\colonequals
    \forall x_1\forall x_2\ldots\forall x_n
    \exists y_1(D_{y_1})\exists y_2(D_{y_2})\ldots\exists y_m(D_{y_m}):
    \varphi
    \label{eq:dqbf_syntax_cp}
  \end{equation*}
  where $D_{y_i}\subseteq\{x_1,\ldots,x_n\}$ for $i=1,\ldots,m$ is
  the \emph{dependency set} of $y_i$, and $\varphi$ is a quantifier-free
  Boolean formula over $V$, called the \emph{matrix} of $\psi$.
\end{definition}

We denote the set of universal variables of $\psi$ by $\varall[\psi]=\{x_1,\ldots,x_n\}$
and its set of existential variables by $\varex[\psi] = \{y_1,\ldots,y_m\}$.
The former part of $\psi$, $\forall x_1\forall x_2\ldots\forall x_n\allowbreak\exists y_1(D_{y_1})\exists y_2(D_{y_2})\ldots\exists y_m(D_{y_m})$,
is called its \emph{prefix}. Sometimes we abbreviate this prefix as $Q$ such that $\psi = Q:\varphi$.

The semantics of closed prenex DQBFs is given as follows:

\begin{definition}[Semantics of closed prenex DQBF]
  \label{def:dqbf_semantics_cp}
  Let $\psi$ be a DQBF with matrix $\varphi$ as above. $\psi$ is
  \emph{satisfiable} iff
  there are functions $s_{y_i}:\assign(D_{y_i})\to\{0,1\}$
  for $1\leq i\leq m$ such that replacing each
  $y_i$ by (a Boolean formula for) $s_{y_i}$ turns
  $\varphi$ into a tautology. Then the functions $(s_{y_i})_{i=1,\ldots,m}$ are called
  \emph{Skolem functions} for $\psi$.
\end{definition}

A DQBF is a QBF, if its dependency sets satisfy certain conditions:
\begin{definition}[Closed prenex QBF]
  \label{def:qbf}
  Let $V = \{x_1,\ldots,x_n,\allowbreak y_1,\ldots,y_m\}$ be a set of Boolean variables.
  A \emph{quantified Boolean formula (QBF)} (more precisely, a closed
  QBF in prenex normal form) $\psi$ over $V$ is given by
  $
    \psi\colonequals\forall X_1 \exists Y_1 \ldots \forall X_k \exists Y_k :\varphi
  $,
  where
  $k\geq 1$, $X_1,\ldots,X_k$ is a partition of the universal variables $\{x_1,\ldots,x_n\}$,
  $Y_1,\ldots,Y_k$ is a partition of the existential variables $\{y_1,\ldots,y_m\}$,
  $X_i\neq\emptyset$ for $i=2,\ldots,k$, and $Y_j\neq\emptyset$ for $j=1,\ldots,k-1$, and
  $\varphi$ is a quantifier-free Boolean formula over $V$.
\end{definition}
A QBF can be seen as a DQBF where the dependency sets are linearly ordered. A QBF
$\psi\colonequals\forall X_1 \exists Y_1 \ldots \forall X_k \exists Y_k : \varphi$
is equivalent to the DQBF
$\psi'\colonequals\forall x_1\ldots\forall x_n\exists y_1(D_{y_1})\ldots$ $\exists y_m(D_{y_m}):\varphi$
with $D_{y_i} = \bigcup_{j=1}^\ell X_j$ where $Y_\ell$ is the unique set with $y_i \in Y_\ell$, $1 \leq \ell \leq k$, $1 \leq i \leq m$.

Quantifier localization for QBF is based on the following theorem (see, \eg\cite{Benedetti05c})
which can be used to transform prenex QBFs into equivalent or equisatisfiable non-prenex QBFs (where the
quantifiers are not necessarily placed before the matrix).
Two QBFs $\psi_1$ and $\psi_2$ are equisatisfiable ($\psi_1 \approx \psi_2$), when $\psi_1$ is satisfiable iff $\psi_2$ is satisfiable.

\begin{theorem}
\label{th:qbfquantifierlocalization}
  Let $\op\in\{{\land}, {\lor}\}$, let $\mathcal{Q} \in \{\exists, \forall\}$, $\overline{\mathcal{Q}} = \exists$,
  if $\mathcal{Q} = \forall$ and $\overline{\mathcal{Q}} = \forall$ otherwise.
  Let $\varfreesupp[\psi]$ be the set of all variables occurring in $\psi$ which are not bound by a quantifier.
  The following holds for all QBFs:
  \begin{small}
    \begin{subequations}
    \begin{align}
       \neg (\mathcal{Q} x : \psi) &\quad\approx\quad  \overline{\mathcal{Q}} x : (\neg \psi)  \label{qbfapprox:neg} \\
       \mathcal{Q} x : \psi &\quad\approx\quad \psi, \;  \text{ if } x \notin \varfreesupp[\psi] \label{qbfapprox:indep} \\
       \forall x : (\psi_1 \land \psi_2) &\quad\approx\quad (\forall x : \psi_1) \land (\forall x : \psi_2)  \label{qbfapprox:forall_and} \\
       \exists x : (\psi_1 \lor \psi_2) &\quad\approx\quad (\exists x : \psi_1) \lor (\exists x : \psi_2)  \label{qbfapprox:exists_or} \\
       \mathcal{Q} x:(\psi_1 \op \psi_2) & \quad\approx\quad \bigl(\psi_1 \op (\mathcal{Q} x:\psi_2)\bigr), \;
           \text{ if } x \notin \varfreesupp[\psi_1] \label{qbfapprox:q_op} \\
       \mathcal{Q} x_1\,\mathcal{Q} x_2:\psi &\quad\approx \quad \mathcal{Q} x_2\,\mathcal{Q} x_1:\psi  \label{qbfapprox:q_q}
    \end{align}
    \end{subequations}
  \end{small}
\end{theorem}

%% file: 03-nonprenex.tex
\section{Non-Closed Non-Prenex DQBFs}
\label{sec:ncnp}

\subsection{Syntax and Semantics}
\noindent In this section, we define syntax and semantics of non-prenex DQBFs.
Since the syntax definition is recursive, we need non-closed DQBFs as well.

\begin{figure*}[tb]
{\footnotesize
  \[
  \renewcommand{\extrarowheight}{1mm}
  \begin{array}{@{}cccccc}
    \toprule
    \textnormal{No.} & \textnormal{Rule} & & \varex[\cdot] & \varall[\cdot] & \varfreesupp[\cdot] \\
    \midrule
    1 &
    \dfrac{v\in V}{v \in \npnc}
        &\quad  & \emptyset & \emptyset & \{v\} \\[3mm]
    2 &
    \dfrac{v\in V}{\neg v \in \npnc}
        && \emptyset & \emptyset & \{v\} \\[2mm]
    3 &
    \dfrac{\varphi_1\in \npnc \qquad\varphi_2\in \npnc \qquad\eqref{eq:disjoint}}{(\varphi_1\land\varphi_2) \in \npnc}
        && \varex[\varphi_1]\dcup\varex[\varphi_2] & \varall[\varphi_1]\dcup\varall[\varphi_2] & \varfreesupp[\varphi_1]\cup\varfreesupp[\varphi_2] \\[4mm]
    4 &
    \dfrac{\varphi_1 \in \npnc \qquad\varphi_2 \in \npnc \qquad\eqref{eq:disjoint}}{(\varphi_1\lor\varphi_2) \in \npnc}
        && \varex[\varphi_1]\dcup\varex[\varphi_2] & \varall[\varphi_1]\dcup\varall[\varphi_2] & \varfreesupp[\varphi_1]\cup\varfreesupp[\varphi_2] \\[5mm]
    5 &
    \dfrac{\varphi \in \npnc \qquad v\in \varfree[\varphi] \qquad D_v\subseteq V\setminus\bigl(\var[\varphi]^Q\dcup\{v\}\bigr)}{\exists v(D_v):\varphi^{-v} \in \npnc}
        && \varex[\varphi]\dcup\{v\} & \varall[\varphi] & \varfreesupp[\varphi]\setminus\{v\} \\[4mm]
    6 &
    \dfrac{\varphi \in \npnc \qquad v\in \varfree[\varphi]}{\forall v:\varphi \in \npnc}
        && \varex[\varphi] & \varall[\varphi]\dcup\{v\} & \varfreesupp[\varphi]\setminus\{v\} \\
    \bottomrule
  \end{array}
  \]
  where \eqref{eq:disjoint} refers to the formula
  \begin{equation}
   \tag{$\ast$}
    \label{eq:disjoint}
      \Big(\var[\varphi_1]^Q \cup \varfreesupp[\varphi_1] \cup \bigcup_{y \in \varex[\varphi_1]} D_y\Big) \cap
      \var[\varphi_2]^Q=\emptyset\quad\land\quad
      \Big(\var[\varphi_2]^Q \cup \varfreesupp[\varphi_2] \cup \bigcup_{y \in \varex[\varphi_2]} D_y\Big) \cap
      \var[\varphi_1]^Q=\emptyset\,.
  \end{equation}
}
  \caption{Rules defining the syntax of non-prenex non-closed DQBFs in negation normal form.\label{fig:rules}}
\end{figure*}

\begin{definition}[Syntax]
  \label{def:syntax}
  Let $V$ be a finite set of Boolean variables. Let $\varphi^{-v}$ result from
  $\varphi$ by removing $v$ from the dependency sets of all existential variables in $\varphi$.

  The set $\npnc$ of \emph{non-closed non-prenex DQBFs in negation normal form} (NNF)
  over $V$, the existential and universal variables as well as the free variables in their
  support are defined by the rules given in Figure~\ref{fig:rules}.
  As usual, $\npnc$ is defined to be the smallest set satisfying those rules.

  $\varex[\psi]$ is the set of existential variables of $\psi$,
  $\varall[\psi]$ the set of universal variables of $\psi$,
  and $\varfreesupp[\psi]$ the set of free variables in the support of $\psi$.
  $\var[\psi] \colonequals \varex[\psi] \dcup \varall[\psi] \dcup \varfreesupp[\psi]$ is
  the set of variables occurring in $\psi$,
  $\var[\psi]^Q \colonequals \varex[\psi]\dcup\varall[\psi]$ is the set of quantified variables of $\psi$,
  and $\varfree[\psi] \colonequals V \setminus \var[\psi]^Q$ is the set of free variables of $\psi$.\footnote{
  In contrast to the variables from $\varfreesupp[\psi]$, the variables from $\varfree[\psi]$ do not necessarily occur in $\psi$.
  }
\end{definition}

\begin{remark}
For the sake of simplicity, we assume in Definition~\ref{def:syntax}
that variables are either free or bound by some quantifier, but not both,
and that no variable is quantified more than once. Every formula that violates this assumption
can easily be brought into the required form by renaming variables.
We restrict ourselves to NNF, since prenex DQBFs are not syntactically closed under negation~\cite{BalabanovCJ14}.
For closed prenex DQBFs the (quantifier-free) matrix can be simply transformed into NNF
by applying De Morgan's rules and omitting double negations (exploiting that $x\equiv \neg\neg x$)
at the cost of a linear blow-up of the formula.
\end{remark}

For two DQBFs $\psi_1,\psi_2$ we write $\psi_1\subformula\psi_2$ if $\psi_1$ is a subformula
of $\psi_2$.

\begin{definition}[Skolem function candidates]
  \label{def:skolem_function_candidates}
  For a DQBF $\psi$ over variables $V$
  in NNF, we define a \emph{Skolem function candidate}
  as a mapping from existential and free variables to functions over universal variables
  $s:\varfree[\psi]\dcup\varex[\psi]\to\boolf{\varall[\psi]}$ with
  \begin{enumerate}
  \item $\support\bigl(s(v)\bigr)=\emptyset$ for all $v\in\varfree[\psi]$, \ie $s(v)\in\{\fzero,\fone\}$, and
  \item $\support\bigl(s(v)\bigr)\subseteq \bigl(D_v\cap\varall[\psi]\bigr)$ for all $v\in\varex[\psi]$.
  \end{enumerate}
  $\sfunc{\psi}$ is the set of all such Skolem function candidates.
\end{definition}

That means, $\sfunc{\psi}$ is the set of all Skolem function candidates satisfying the
constraints imposed by the dependency sets of the existential and free variables.

\begin{notation}
\label{not:skolem}
Given $s\in\sfunc{\psi}$ for a DQBF $\psi\in\npnc$, we write
$s(\psi)$ for the formula that results from $\psi$ by replacing each variable $v$ for which $s$ is defined
by $s(v)$ and omitting all quantifiers from $\psi$, \ie $s(\psi)$ is a quantifier-free Boolean formula,
containing only variables from $\varall[\psi]$.
\end{notation}

\begin{definition}[Semantics of DQBFs in NNF]
  \label{def:sem}
  Let $\psi\in\npnc$ be a DQBF over variables $V$. We define the semantics $\sem{\psi}$ of $\psi$ as follows:
  \[
    \sem{\psi} \colonequals
      \bigl\{s\in\sfunc{\psi}\,\big|\,\vDash s(\psi)\bigr\}
      =
      \bigl\{s\in\sfunc{\psi}\,\big|\,\forall\mu\in\assign(\varall[\psi]): \mu\bigl(s(\psi)\bigr) = 1\bigr\}.
   \]
  $\psi$ is \emph{satisfiable} if $\sem{\psi}\neq\emptyset$; otherwise we call it \emph{unsatisfiable}.
  The elements of $\sem{\psi}$ are called \emph{Skolem functions} for $\psi$.
\end{definition}
The semantics $\sem{\psi}$ of $\psi$ is the subset of $\,\sfunc{\psi}$ such that for all $s\in\sem{\psi}$ we have:
Replacing each free or existential variable $v\in\varfree[\psi]\dcup\varex[\psi]$ with a Boolean expression
for $s(v)$ turns $\psi$ into a tautology.

\begin{example}
\label{ex1}
Consider the DQBF
\[
  \psi\colonequals \forall x_1 \forall x_2  : \Bigl(
    \bigl(x_1\equiv x_2\bigr) \vee \bigl(\exists y_1(x_2) : (x_1\not\equiv y_1)\bigr)
    \Bigr)
\]
over the set of variables $\{x_1, x_2, y_1\}$.
$y_1$ with dependency set $\{x_2\}$ is the only existential variable in $\psi$ and
there are no free variables. Thus
$\sfunc{\psi} = \{y_1 \mapsto \fzero, y_1 \mapsto \fone, y_1 \mapsto x_2, y_1 \mapsto \neg x_2\}$.
It is easy to see that $s \colonequals y_1 \mapsto x_2$ is a Skolem function for $\psi$, since
$\vDash s(\psi) = \bigl((x_1\equiv x_2) \vee (x_1\not\equiv x_2)\bigr)$, and that the other Skolem function
candidates do not define Skolem functions.
\end{example}

\begin{remark}
   For closed prenex DQBFs the semantics defined here obviously coincides with the usual semantics
   as specified in Definition~\ref{def:dqbf_semantics_cp} if we transform the (quantifier-free)
   matrix into NNF first.
\end{remark}

\begin{remark}
  A (non-prenex) DQBF $\psi$ is a (non-prenex) QBF if every existential variable depends on
  all universal variables in whose scope it is (and possibly on free variables as well).
\end{remark}

The following theorem provides a constructive characterization of the semantics of a DQBF.
\begin{restatable}{theorem}{semantics}
  \label{th:semantics}
  The set $\sem{\psi}$ for a DQBF $\psi$ over variables $V$ in NNF can be characterized recursively as follows:
  \begin{subequations}
    {\small\allowdisplaybreaks
    \begin{align}
      \sem{v} &= \bigl\{ s\in\sfunc{v}\,\big|\,s(v)=\fone \bigr\} \text{ for } v\in V_\psi,
          \label{th:semantics:var} \\
      \sem{\neg v} &= \bigl\{ s\in\sfunc{\neg v}\,\big|\,s(v) = \fzero\bigr\} \text{ for } v\in V_\psi,
          \label{th:semantics:notvar} \\
      \sem{(\varphi_1\land\varphi_2)} &= \bigl\{ s\in\sfunc{\varphi_1\land\varphi_2}\,\big|\,  \label{th:semantics:and}
       s_{|\varfree[\varphi_1]\dcup\varex[\varphi_1]}\in\sem{\varphi_1} \land s_{|\varfree[\varphi_2]\dcup\varex[\varphi_2]}\in\sem{\varphi_2}\bigr\},
          \\
      \sem{(\varphi_1\lor\varphi_2)} &= \bigl\{ s\in\sfunc{\varphi_1\lor\varphi_2}\,\big|\, \label{th:semantics:or}
       s_{|\varfree[\varphi_1]\dcup\varex[\varphi_1]}\in\sem{\varphi_1} \lor s_{|\varfree[\varphi_2]\dcup\varex[\varphi_2]}\in\sem{\varphi_2}\bigr\},
          \\
      \sem{\exists v(D_v):\varphi^{-v}} &= \sem{\varphi^{-v}},
          \label{th:semantics:exists} \\
      \sem{\forall v:\varphi} &= \Bigl\{ t\in\sfunc{\forall v:\varphi}\,\Big|\, \label{th:semantics:forall}
            \exists s_0, s_1\in\sem{\varphi}:  s_0(v)=\fzero\land s_1(v)=\fone \; \land  \\*
           & \qquad \forall w\in\varfree[\forall v:\varphi]: t(w) = s_0(w) = s_1(w) \; \land \nonumber \\*
           & \qquad \forall w\in\varex[\forall v:\varphi], v\notin D_w : t(w) = s_0(w) = s_1(w) \; \land  \nonumber \\*
           & \qquad \forall w\in\varex[\forall v:\varphi], v\in D_w : t(w) = \mathrm{ITE}\bigl(v,\ s_1(w),\ s_0(w)\bigr) \Bigr\}\nonumber
  \end{align}}
\end{subequations}
\end{restatable}

For the proof as well as for the following example, we denote the semantics defined in Definition~\ref{def:sem}
by $\semdef{\psi}$ (\ie $\semdef{\psi}=\{ s\in\sfunc{\psi}\,|\,\vDash s(\psi)\}$)
and the set that is characterized by Theorem~\ref{th:semantics} by
$\semth{\psi}$.

\begin{proof}
  $\semdef{\psi}=\semth{\psi}$ is shown by induction on the structure of $\psi$,
  for details we refer to \ref{app:semproof}.
\end{proof}

The following example illustrates the recursive characterization of Theorem~\ref{th:semantics} (and again
the recursive Definition~\ref{def:syntax}).

\begin{example}
\label{ex2}
Let us consider the DQBF
\[
  \psi\colonequals \forall x_1 \forall x_2  : \Bigl(
    \bigl(x_1\equiv x_2\bigr) \lor \bigl(\exists y_1(x_2) : (x_1\not\equiv y_1)\bigr)
    \Bigr)
\]
over the set of variables $\{x_1, x_2, y_1\}$ from Example~\ref{ex1} again.
We compute $\semth{\psi}$ recursively.

As an abbreviation for
$\bigl((\neg x_1 \land y_1) \lor (x_1 \land \neg y_1)\bigr)$,
$(x_1\not\equiv y_1)$ is a DQBF
based on rules 1--4 of Definition~\ref{def:syntax} with
$\varex[x_1\not\equiv y_1] = \varall[x_1\not\equiv y_1] = \emptyset$,
$\varfreesupp[x_1\not\equiv y_1] = \{x_1, y_1\}$, $\varfree[x_1\not\equiv y_1] = \{x_1, x_2, y_1\}$.
With Theorem~\ref{th:semantics}, \eqref{th:semantics:var}--\eqref{th:semantics:or}
we get
$\semth{x_1\not\equiv y_1} = \bigl\{s : \{x_1, x_2, y_1\} \to\boolf{\emptyset} \; \big| \; s(y_1) \neq s(x_1)\bigr\}$.

For $\psi' \colonequals \bigl(\exists y_1(x_2) : (x_1\not\equiv y_1)\bigr)$, we obtain by rule 5:
$\varall[\psi'] = \emptyset$,
$\varex[\psi'] = \{y_1\}$,
$\varfreesupp[\psi'] = \{x_1\}$.
$\varfree[\psi'] = \{x_1, x_2\}$.
According to Theorem~\ref{th:semantics}, \eqref{th:semantics:exists} we have
$\semth{\exists y_1(x_2) : (x_1\not\equiv y_1)} = \semth{x_1\not\equiv y_1}$.

Similarly we obtain
$\semth{x_1\equiv x_2} = \bigl\{s : \{x_1, x_2, y_1\} \to\boolf{\emptyset} \; \big| \; s(x_1) = s(x_2)\bigr\}$.

Then, for $\psi'' \colonequals \bigl(\bigl(x_1\equiv x_2\bigl) \vee \bigl(\exists y_1(x_2) : (x_1\not\equiv y_1)\bigr)\bigr)$
we have
$\varall[\psi''] = \emptyset$,
$\varex[\psi''] = \{y_1\}$,
$\varfreesupp[\psi''] = \varfree[\psi''] = \{x_1, x_2\}$, and
by Theorem~\ref{th:semantics}, \eqref{th:semantics:or}
$\semth{\psi''} = \bigl\{s : \{x_1, x_2, y_1\} \to\boolf{\emptyset}  \; \bigr| \; \bigl(s(x_1), s(x_2), s(y_1)\bigr) \in
\{(\fzero, \fzero, \fzero), (\fzero, \fzero, \fone), (\fzero, \fone, \fone),
(\fone, \fzero, \fzero), (\fone, \fone, \fzero), (\fone, \fone, \fone) \}
\bigr\}$.

Now we consider $\forall x_2 : \psi''$.
$\varall[\forall x_2 :  \psi''] = \{x_2\}$.
$\varex[\forall x_2 :  \psi''] = \{y_1\}$,
$\varfreesupp[\forall x_2 :  \psi''] = \varfree[\forall x_2 :  \psi''] = \{x_1\}$.
We use $\eqref{th:semantics:forall}$
to construct $\semth{\forall x_2 : \psi''}$.
In principle, there are three possible choices $s_0 \in \semth{\psi''}$ with $s_0(x_2) = \fzero$
and three possible choices $s_1 \in \semth{\psi''}$ with $s_1(x_2) = \fone$.
Due to the constraint $s_0(x_1) = s_1(x_1)$ in the third line of $\eqref{th:semantics:forall}$,
there remain only four possible combinations $s^{(1)}_0, s^{(1)}_1, \ldots, s^{(4)}_0, s^{(4)}_1$:
\begin{itemize}
\item $\bigl(s^{(1)}_0(x_1), s^{(1)}_0(x_2), s^{(1)}_0(y_1)\bigr) = (\fzero, \fzero, \fzero)$, \\
      $\bigl(s^{(1)}_1(x_1), s^{(1)}_1(x_2), s^{(1)}_1(y_1)\bigr) = (\fzero, \fone, \fone)$, leading to  \\
      $t^{(1)}(x_1) = \fzero$, $t^{(1)}(y_1) = \mathrm{ITE}\bigl(x_2, s^{(1)}_1(y_1), s^{(1)}_0(y_1)\bigr) = \mathrm{ITE}(x_2, \fone, \fzero) = x_2$,
\item $\bigl(s^{(2)}_0(x_1), s^{(2)}_0(x_2), s^{(2)}_0(y_1)\bigr) = (\fzero, \fzero, \fone)$, \\
      $\bigl(s^{(2)}_1(x_1), s^{(2)}_1(x_2), s^{(2)}_1(y_1)\bigr) = (\fzero, \fone, \fone)$, leading to  \\
      $t^{(2)}(x_1) = \fzero$, $t^{(2)}(y_1) = \mathrm{ITE}\bigl(x_2, s^{(2)}_1(y_1), s^{(2)}_0(y_1)\bigr) = \mathrm{ITE}(x_2, \fone, \fone) = \fone$,
\item $\bigl(s^{(3)}_0(x_1), s^{(3)}_0(x_2), s^{(3)}_0(y_1)\bigr) = (\fone, \fzero, \fzero)$, \\
      $\bigl(s^{(3)}_1(x_1), s^{(3)}_1(x_2), s^{(3)}_1(y_1)\bigr) = (\fone, \fone, \fzero)$, leading to  \\
      $t^{(3)}(x_1) = \fone$, $t^{(3)}(y_1) = \mathrm{ITE}\bigl(x_2, s^{(3)}_1(y_1), s^{(3)}_0(y_1)\bigr) = \mathrm{ITE}(x_2, \fzero, \fzero) = \fzero$,
\item $\bigl(s^{(4)}_0(x_1), s^{(4)}_0(x_2), s^{(4)}_0(y_1)\bigr) = (\fone, \fzero, \fzero)$, \\
      $\bigl(s^{(4)}_1(x_1), s^{(4)}_1(x_2), s^{(4)}_1(y_1)\bigr) = (\fone, \fone, \fone)$, leading to  \\
      $t^{(4)}(x_1) = \fone$, $t^{(4)}(y_1) = \mathrm{ITE}\bigl(x_2, s^{(4)}_1(y_1), s^{(4)}_0(y_1)\bigr) = \mathrm{ITE}(x_2, \fone, \fzero) = x_2$.
\end{itemize}
Altogether, $\semth{\forall x_2 : \psi''} = \{t^{(1)}, t^{(2)}, t^{(3)}, t^{(4)}\}$.

Finally, for $\psi = \forall x_1 \forall x_2  : \psi''$ we have
$\varall[\psi] = \{x_1, x_2\}$,
$\varex[\psi] = \{y_1\}$,
$\varfreesupp[\psi] = \varfree[\psi] = \emptyset$.
For the choice of $s_0$ and $s_1$ in \eqref{th:semantics:forall} we need
$s_0(x_1) = \fzero$, $s_1(x_1) = \fone$ and, due to
$x_1 \notin D_{y_1}$, $s_0(y_1) = s_1(y_1)$ (see the third line of \eqref{th:semantics:forall}).
Thus, the only possible choice is $s_0 \colonequals t^{(1)}$ and $s_1 \colonequals t^{(4)}$;
$t(y_1) \colonequals x_2$ is the only possible Skolem function for $\psi$.
This result agrees with the Skolem function computed using
Definition~\ref{def:sem} in Example~\ref{ex1}.
\end{example}

\subsection{Equivalent and Equisatisfiable Non-Closed Non-Prenex DQBFs}

\noindent Now we define rules for replacing DQBFs by equivalent and equisatisfiable ones.
We start with the definition of equivalence and equisatisfiability:
\begin{definition}[Equivalence and equisatisfiability]
  \label{def:equivalence}
  Let $\psi_1,\psi_2\in\npnc$ be DQBFs over $V$. We call them \emph{equivalent} (written $\psi_1\equiv\psi_2$)
  if $\sem{\psi_1} = \sem{\psi_2}$; they are
  \emph{equisatisfiable} (written $\psi_1\approx\psi_1$)
  if $\sem{\psi_1} = \emptyset\Leftrightarrow\sem{\psi_2} = \emptyset$ holds.
\end{definition}

Now we prove Theorem~\ref{th:rules}, which is the DQBF counterpart to Theorem~\ref{th:qbfquantifierlocalization} for QBF.

\begin{restatable}{theorem}{rules}
  \label{th:rules}
  Let $\op\in\{{\land}, {\lor}\}$
  and all formulas occurring on the left- and right-hand sides of the following rules
  be DQBFs in $\npnc$ over the same set $V$ of variables.
  We assume that $x'$ and $y'$ are fresh variables,
  which do not occur in $\varphi$, $\varphi_1$, and $\varphi_2$.
  The following equivalences and equisatifiabilities hold for all DQBFs in NNF.

  \begin{small}
    \begin{subequations}
    {\allowdisplaybreaks\begin{align}
       \exists y(D_y): \varphi &\quad\equiv\quad \varphi
          \label{equiv:indep_exists} \\
       \forall x : \varphi &\quad\approx\quad \varphi^{-x}
          \quad\text{ if } x \notin \var[\varphi]
          \label{equiv:indep} \\
       \forall x : \varphi &\quad\equiv\quad \varphi[\sfrac{0}{x}] \land \varphi[\sfrac{1}{x}]
          \quad\text{ if } \varall[\varphi] = \varex[\varphi] = \emptyset
          \label{th:rules3} \\
       \exists y(D_y) : \varphi &\quad \approx \quad \varphi[\sfrac{0}{y}] \lor \varphi[\sfrac{1}{y}]
          \quad\text{ if } \varall[\varphi] = \varex[\varphi] = \emptyset\footnotemark
          \label{equiv:exists} \\
       \forall x:(\varphi_1\land\varphi_2) &\quad\approx\quad \bigl(\forall x:\varphi_1\bigr)\land\bigl(\forall x':\varphi_2[\sfrac{x'}{x}]\bigr)\footnotemark
          \label{equiv:forall_and} \\
       \forall x:(\varphi_1 \land \varphi_2) & \quad \approx \quad \bigl(\varphi_1^{-x} \land (\forall x:\varphi_2)\bigr)
          \quad\text{ if } x\notin\var[\varphi_1]
          \label{equiv:forall_and2} \\
       \forall x:(\varphi_1 \op \varphi_2) & \quad\equiv\quad \bigl(\varphi_1 \op (\forall x:\varphi_2)\bigr)
          \quad\text{ if } x\notin\var[\varphi_1] \text{ and } x\notin D_y \text{ for all } y\in\varex[\varphi_1]
          \label{equiv:forall_or} \\
       \exists y(D_y):(\varphi_1\lor\varphi_2) &\quad\approx\quad \bigl(\exists y(D_y):\varphi_1\bigr)
          \lor\bigl(\exists y'(D_y):\varphi_2[\sfrac{y'}{y}]\bigr)
          \label{equiv:exists_or} \\
       \exists y(D_y):(\varphi_1\op\varphi_2) &\quad\equiv\quad \bigl(\varphi_1\op(\exists y(D_y):\varphi_2)\bigr)
          \quad\text{ if } y\notin\var[\varphi_1]
          \label{equiv:exists_and} \\
       \exists y_1(D_{y_1})\,\exists y_2(D_{y_2}):\varphi &\quad\equiv\quad \exists y_2(D_{y_2})\,\exists y_1(D_{y_1}):\varphi
          \label{equiv:exists_exists} \\
       \forall x_1\,\forall x_2:\varphi &\quad\equiv\quad \forall x_2\,\forall x_1:\varphi
          \label{equiv:forall_forall} \\
       \forall x\,\exists y(D_y):\varphi &\quad\equiv\quad \exists y(D_y)\,\forall x:\varphi
          \quad\text{ if } x\notin D_y.
          \label{equiv:forall_exists}
    \end{align}}
    \end{subequations}
  \end{small}
\end{restatable}
\addtocounter{footnote}{-1}
\footnotetext{A generalization of this rule (which is not needed in this paper, however) holds also for the case that $\varall[\varphi] = \varex[\varphi] = \emptyset$ does not hold. In this case the quantified variables in
$\varphi[\sfrac{0}{y}]$ or in $\varphi[\sfrac{1}{y}]$ have to be renamed in order to satisfy the conditions
of Definition~\ref{def:syntax}.}
\addtocounter{footnote}{1}
\footnotetext{By $\varphi_2[\sfrac{x'}{x}]$ we mean that all occurrences of $x$ are replaced by $x'$, including the occurrences in dependency sets.}

Note that the duality of $\exists$ and $\forall$ under negation as in QBF
($\exists \varphi \equiv \neg\forall\neg\varphi$) does not
hold for DQBF as DQBFs are not syntactically closed under negation~\cite{BalabanovCJ14}.

\begin{example}
\label{ex3}
We give an example that shows that
--~in contrast to \eqref{qbfapprox:q_op} of Theorem~\ref{th:qbfquantifierlocalization} for QBF~--
the condition
$x\notin D_y$ for all $y\in\varex[\varphi_1]$
is really needed in \eqref{equiv:forall_or} if $\op = \lor$.
We consider the satisfiable DQBF
$\psi \colonequals \forall x_1 \forall x_2 : \bigl((x_1\equiv x_2) \vee (\exists y_1(x_2) : (x_1\not\equiv y_1))\bigr)$
from Example~\ref{ex1} again.
First of all, neglecting the above condition, we could transform $\psi$ into
$\psi' \colonequals \forall x_1 : \bigl((\forall x_2 : (x_1\equiv x_2)) \vee (\exists y_1(x_2) : (x_1\not\equiv y_1))\bigr)$,
which is not well-formed according to Definition~\ref{def:syntax}.
However, by renaming $x_2$ into $x'_2$ in the dependency set of $y_1$ we would arrive
at a well-formed DQBF
$\psi'' \colonequals \forall x_1  : \bigl((\forall x_2 : (x_1\equiv x_2)) \vee (\exists y_1(x'_2) : (x_1\not\equiv y_1))\bigr)$.
According to Definition~\ref{def:skolem_function_candidates} the only possible Skolem function
candidates for $y_1$ in $\psi''$ are $\fzero$ and $\fone$.
It is easy to see that neither inserting $\fzero$ nor $\fone$ for $y_1$ turns
$\psi''$ into a tautology, thus $\psi''$ is unsatisfiable and therefore
neither equivalent to nor equisatisfiable with $\psi$.
\end{example}

Whereas the proof of Theorem~\ref{th:qbfquantifierlocalization} for QBF is rather easy using
the equisatisfiabilities $\exists y : \varphi \approx \varphi[\sfrac{0}{y}] \lor \varphi[\sfrac{1}{y}]$
and $\forall x : \varphi \approx \varphi[\sfrac{0}{x}] \land \varphi[\sfrac{1}{x}]$, the proof of Theorem~\ref{th:rules}
is more involved. We provide a detailed proof in \ref{app:ruleproof}.

It is also important to note that some of the rules in Theorem~\ref{th:rules}
establish equivalences and some establish equisatisfiabilities only.
Whereas this might seem to be negligible if we are only interested in the
question whether a formula is satisfiable or not, it turns out to be essential
in the context of Section~\ref{sec:algorithm} where we replace \emph{subformulas} by
equivalent or equisatisfiable formulas. Replacing a subformula of a formula $\psi$
by an equisatisfiable subformula does not necessarily preserve satisfiability / unsatisfiability.
This observation is trivially true already for pure propositional logic
(\eg $x_1$ is equisatisfiable with $x_1 \lor x_2$, but
$x_1 \land \neg x_1$ is \emph{not} equisatisfiable with $(x_1 \lor x_2) \land \neg x_1$).
Here we show a more complex example for DQBFs:
\begin{example}[label=ex:counter]
Let us consider the DQBF
$\psi \colonequals \exists y(\emptyset): \bigl((x \land y) \lor (\neg x \land \neg y)\bigr)$,
which is, according to \eqref{equiv:exists_or}, equisatisfiable with
$\psi' \colonequals \bigl(\exists y(\emptyset) : (x \land y)\bigr)
\lor \bigl(\exists y'(\emptyset): (\neg x \land \neg y')\bigr)$.
$\forall x : \psi$
is unsatisfiable, since for the choice $s(y) = \fzero$ we have $s(\forall x : \psi) \equiv \neg x$
and for the choice $s(y) = \fone$ we have $s(\forall x : \psi) \equiv x$, \ie for both possible choices for the
Skolem function candidates we do not obtain a tautology.
However,
$\forall x : \psi'$ is satisfiable by $s(y) = \fone$ and $s(y') = \fzero$.
\end{example}

It is easy to see that situations like in Example~\ref{ex:counter}
do not occur when we replace \emph{equivalent} subformulas:

\begin{theorem}
\label{th:equiv_subformulas}
Let $\psi$ be a DQBF and $\psi_1 \subformula \psi$ be a subformula of $\psi$.
Let $\psi_2$ be a DQBF that is equivalent to $\psi_1$,
$\varex[\psi_1] = \varex[\psi_2]$, and
each existential variable $y$ has the same dependency set in $\psi_2$ as in $\psi_1$.
Then the DQBF $\psi' \colonequals \psi[\sfrac{\psi_2}{\psi_1}]$, which results from replacing $\psi_1$ by $\psi_2$,
is equivalent to $\psi$.
\end{theorem}

\begin{proofsketch}
The proof easily follows from the fact that the set of Skolem functions is identical
for equivalent subformulas and from the recursive characterization of the semantics of DQBFs in
Theorem~\ref{th:semantics}.
If the existential variables in $\psi_1$ and $\psi_2$ as well as their dependency sets
are identical, then the same holds
for all $\psi'_1 \sqsubseteq \psi$ and $\psi'_2 \sqsubseteq \psi'$ with
$\psi_1' \not\sqsubseteq \psi_1$ and $\psi_2' \not\sqsubseteq \psi_2$.
We make use of this condition in part \eqref{th:semantics:forall} of
the inductive proof that shows $\psi \equiv \psi'$.
Assume $\psi_1\ \subformula\ \forall v:\varphi$.
In part \eqref{th:semantics:forall}, free and existential variables
of $\forall v:\varphi$
(resp.~of $\forall v:\varphi[\sfrac{\psi_2}{\psi_1}]$)
are handled
differently and therefore it is not enough to inductively assume that the Skolem functions
for $\varphi$ and $\varphi[\sfrac{\psi_2}{\psi_1}]$
are identical, but existential / free variables should also not have ``changed their type''.
Moreover, existential variables with different dependency sets are handled differently (see last two lines
of \eqref{th:semantics:forall}), so we also have to assume that the dependency set of each existential variable $y$
in $\psi_1$ is the same as its dependency set in $\psi_2$.
\qed
\end{proofsketch}

\begin{remark}
Theorem~\ref{th:equiv_subformulas} says that it is safe to
apply rules \eqref{th:rules3}, \eqref{equiv:forall_or}, \eqref{equiv:exists_and},
\eqref{equiv:exists_exists}, \eqref{equiv:forall_forall}, and \eqref{equiv:forall_exists}
to subformulas as well.
Theorem~\ref{th:equiv_subformulas} does not imply this for rule~\eqref{equiv:indep_exists},
since the sets of existential variables of the left-hand side formula and the
right-hand side formula are different.
\end{remark}

Since we are still interested in obtaining equisatisfiable formulas by
replacing equisatisfiable subformulas, we need to have a closer look at the
rules \eqref{equiv:indep_exists}, \eqref{equiv:indep}, \eqref{equiv:exists}, \eqref{equiv:forall_and}, \eqref{equiv:forall_and2}
and \eqref{equiv:exists_or}. Example~\ref{ex:counter} already shows that
we will not be able to achieve our goal in all cases without considering additional
conditions.

We start with rule \eqref{equiv:indep_exists}.
With the restriction $y\notin\var[\varphi]$ rule~\eqref{equiv:indep_exists}
can simply be generalized to subformulas:
\begin{theorem}
\label{th:equiv:indep_exists}
  Let $\psi \in \npnc$ be a DQBF over $V$
  and let $\exists y(D_y): \varphi\ \subformula\ \psi$ be a subformula of $\psi$
  with $y\notin\var[\varphi]$.
  Then $\psi \approx \psi'$ where $\psi'$ results from $\psi$ by replacing the subformula
  $\exists y(D_y): \varphi$ by $\varphi$.
\end{theorem}
\begin{proof}
For each $s \in \sem{\psi'}$ we have $s \in \sem{\psi}$
and $\vDash s(\psi')$ implies $\vDash s(\psi)$.

Now assume $\sem{\psi} \neq \emptyset$ and $s \in \sem{\psi}$ with $\vDash s(\psi)$.
Define $s' \in \sem{\psi'}$ by $s'(y) \colonequals c$ for some $c \in \{\fzero, \fone\}$ and
$s'(v) \colonequals s(v)$ for $v\neq y$. Since $y\notin\var[\varphi]$,
the only occurrence of $y$ in $\psi$ is in $\exists y(D_y)$
due to the rules in Definition~\ref{def:syntax}
and there is no occurrence of $y$ in $\psi'$, \ie $s(\psi) = s'(\psi')$.
Thus $\vDash s(\psi)$ implies $\vDash s'(\psi')$.
\end{proof}

Next we consider rule \eqref{equiv:indep}. Although this rule establishes equisatisfiability only,
it may be generalized to the replacement of subformulas:

\begin{theorem}
\label{th:equiv:indep}
  Let $\psi \in \npnc$ be a DQBF
  and let $\forall x : \varphi\ \subformula\ \psi$ be a subformula of $\psi$
  with $x\notin\var[\varphi]$.
  Then $\psi \approx \psi'$ where $\psi'$ results from $\psi$ by replacing the subformula
  $\forall x : \varphi$ by $\varphi^{-x}$.
\end{theorem}
\begin{proof}
  If $\sem{\psi'} \neq \emptyset$, then for each $s \in \sem{\psi'}$
  with $\vDash s(\psi')$ we also have $\vDash s(\psi)$ and $s \in \sem{\psi}$.
  $s$ is a Skolem function candidate for $\psi$ as well, since the supports
  of the Skolem function candidates for $\psi$ cannot be more restricted than
  those for $\psi'$.

  Now assume $\sem{\psi} \neq \emptyset$ and $s \in \sem{\psi}$.
  Choose $s'$ with $s'(v) \colonequals s(v)[\sfrac{c}{x}]$ for all existential
  variables $v \in \varex[\varphi]$ for an arbitrary constant $c \in \{0, 1\}$
  and $s'(v) \colonequals s(v)$ otherwise.
  It is clear that $s'$ is a Skolem function candidate for $\psi'$.
  Now consider an arbitrary assignment $\mu' \in \assign(\varall[\psi'])$.
  Choose $\mu \in \assign(\varall[\psi])$ by $\mu(v) \colonequals \mu'(v)$ for all
  $v \in \varall[\psi']$ and $\mu(x) \colonequals c$.
  Because of $x\notin\var[\varphi]$, $x$ occurs in $s(\psi)$ only in Skolem functions.
  Therefore $\mu'\bigl(s'(\psi')\bigr) = \mu\bigl(s(\psi)\bigr)$ and $\mu'\bigl(s'(\psi')\bigr) = 1$,
  since $s(\psi)$ is a tautology.
  This proves that $s'(\psi')$ is a tautology as well and $\sem{\psi'} \neq \emptyset$.
\end{proof}

Next, we look into rule \eqref{equiv:exists}. Theorem~\ref{th:rules2}
is a variant of this rule which is suitable for replacements of subformulas.
Here we have the first situation that we need additional conditions
for the proof to go through in the more general context of replacing subformulas.
Theorem~\ref{th:rules2} is strongly needed for our
algorithm taking advantage of quantifier localization.
It shows that, under certain conditions, we can do symbolic quantifier
elimination for non-prenex DQBFs as it is known from QBFs:
\begin{restatable}{theorem}{rulesTwo}
  \label{th:rules2}
  Let $\psi \in \npnc$ be a DQBF and let $\exists y(D_y): \varphi\ \subformula\ \psi$ be a subformula of $\psi$
  such that $\varall[\varphi] = \varex[\varphi] = \emptyset$
  and
  $\var[\varphi] \subseteq D_y \cup \varfree[\psi] \cup \{v \in \varex[\psi] \; | \; D_v \subseteq D_y\}$.
  Then $\psi \approx \psi'$ where $\psi'$ results from $\psi$ by replacing the subformula
  $\exists y(D_y): \varphi$ by $\varphi[\sfrac{0}{y}] \lor \varphi[\sfrac{1}{y}]$.
\end{restatable}
The proof of Theorem~\ref{th:rules2} is somewhat involved and uses results from \cite{Jiang09}.
\begin{proofsketch}
We show equisatisfiability by proving that $\sem{\psi'} \neq \emptyset$ implies
$\sem{\psi} \neq \emptyset$ and vice versa.

First assume that there is a Skolem function $s' \in  \sem{\psi'}$ with $\vDash s'(\psi')$.
We define $s \in \sfunc{\psi}$ by $s(v) \colonequals s'(v)$ for all $v \in \varex[\psi'] \cup \varfree[\psi'] \setminus \{y\}$
and $s(y) \colonequals s'(\varphi[\sfrac{1}{y}])$.
The fact that $s \in \sfunc{\psi}$ follows from the restriction that
$\varphi$ contains only variables from $D_y \cup \varfree[\psi] \cup \{v \in \varex[\psi] \; | \; D_v \subseteq D_y\}$, \ie
$\support\bigl(s(y)\bigr) = \support\bigl(s'(\varphi[\sfrac{1}{y}])\bigr) \subseteq D_y$.
$\vDash s(\psi)$ follows by some rewriting from a result in \cite{Jiang09} proving that
quantifier elimination can be done by composition, \ie
$\varphi[\sfrac{\varphi[\sfrac{1}{y}]}{y}]$ is equivalent to $\varphi[\sfrac{0}{y}] \lor \varphi[\sfrac{1}{y}]$.

Now assume $s \in  \sem{\psi}$ with $\vDash s(\psi)$ and define
$s'$ just by removing $y$ from the domain of $s$.
In a first step we change $s$ into $s''$ by replacing $s(y)$ with
$s''(y) \colonequals s'(\varphi)[\sfrac{1}{y}]$.
We conclude
$\vDash s''(\psi)$ from \cite{Jiang09} and monotonicity properties
of $\psi$ in negation normal form. In a second step we use
\cite{Jiang09} again to show that $s''(\psi)$ is equivalent to
$s'(\psi')$. Thus finally $\vDash s'(\psi')$.
The detailed proof can be found in \ref{app:rules2}.\qed
\end{proofsketch}

The generalization of rule \eqref{equiv:forall_and} to replacements of subformulas is formulated
in Theorem~\ref{th:forall_and}. Here we do not need any additional conditions:

\begin{restatable}{theorem}{forallAnd}
  \label{th:forall_and}
  Let $\psi \in \npnc$ be a DQBF and let $\psi_1\colonequals \forall x:(\varphi_1\land\varphi_2)\ \subformula\ \psi$ be a subformula of $\psi$.
  Then $\psi \approx \psi'$ where $\psi'$ results from $\psi$ by replacing the subformula
  $\psi_1$ by $\psi_2\colonequals \bigl(\forall x:\varphi_1\bigr)\land\bigl(\forall x':\varphi_2[\sfrac{x'}{x}]\bigr)$.
\end{restatable}

Before proving Theorem~\ref{th:forall_and}, we consider a lemma which will
be helpful for the proofs of Theorems~\ref{th:forall_and}, \ref{th:equisat:forall_and},
and \ref{th:equisat:exists_or}.

\begin{lemma}\label{lemma:monotonic}
  Let $\varphi_1 \subformula \varphi$ be quantifier-free Boolean formulas
  such that $\varphi_1$ is not in the scope of any negation from $\varphi$.
  Let $\varphi' \colonequals \varphi[\sfrac{\varphi_2}{\varphi_1}]$  be
  the Boolean formula resulting from the replacement of $\varphi_1$ by the
  quantifier-free Boolean formula $\varphi_2$ and let
  $\mu \in \assign(\var[\varphi] \cup \var[\varphi'])$.
  If $\mu(\varphi) = 1$ and $\mu(\varphi') = 0$, then $\mu(\varphi_1) = 1$ and
  $\mu(\varphi_2) = 0$.
\end{lemma}

\begin{proof}
  By assumptions, $\varphi_1$ is only in the scope of conjunctions and disjunctions.
  Due to monotonicity of conjunctions and disjunctions we have
  $\mu(\varphi[\sfrac{0}{\varphi_1}])\leq\mu(\varphi[\sfrac{1}{\varphi_1}])$.

  Moreover, by construction, we have
  $\varphi'[\sfrac{0}{\varphi_2}] = \varphi[\sfrac{0}{\varphi_1}]$ and
  $\varphi'[\sfrac{1}{\varphi_2}]) = \varphi[\sfrac{1}{\varphi_1}]$.
  Further, $\mu(\varphi) = \mu(\varphi[\sfrac{c}{\varphi_1}])$ if
  $\mu(\varphi_1) = c \in \{0, 1\}$, and
  $\mu(\varphi') = \mu(\varphi'[\sfrac{c}{\varphi_2}])$ if
  $\mu(\varphi_2) = c \in \{0, 1\}$.

  $\mu(\varphi_1) = \mu(\varphi_2)$ is not possible, since we assumed
  $\mu(\varphi) = 1$ and $\mu(\varphi') = 0$, \ie
  $\mu(\varphi) \neq \mu(\varphi')$.
  From $\mu(\varphi_1) = 0$ and $\mu(\varphi_2) = 1$ we could conclude
  $\mu(\varphi) = \mu(\varphi[\sfrac{0}{\varphi_1}]) \leq
  \mu(\varphi[\sfrac{1}{\varphi_1}]) =
  \mu(\varphi'[\sfrac{1}{\varphi_2}]) =
  \mu(\varphi')$, which also contradicts $\mu(\varphi) = 1$ and $\mu(\varphi') = 0$.
\end{proof}

\begin{proofsketch}[Theorem~\ref{th:forall_and}]
The proof of Theorem~\ref{th:forall_and} uses Lemma~\ref{lemma:monotonic} to lift
rule~\eqref{equiv:forall_and} to the more general case of replacements of subformulas.
It is easy, but rather technical, and can be found in \ref{app:forall_and}.\qed
\end{proofsketch}

The next theorem shows that rule \eqref{equiv:forall_and2} can also be generalized to replacements of
subformulas without needing additional conditions.
\begin{theorem}
  \label{th:equisat:forall_and}
  Let $\psi \in \npnc$ be a DQBF and let $\psi_1 \colonequals
  \forall x:(\varphi_1 \land \varphi_2)\ \subformula\ \psi$ be a subformula of $\psi$
  with $x\notin\var[\varphi_1]$.
  Then $\psi \approx \psi'$ where $\psi'$ results from $\psi$ by replacing the subformula
  $\psi_1$ by $\psi_2 \colonequals \bigl(\varphi_1^{-x} \land (\forall x:\varphi_2)\bigr)$.
\end{theorem}

\begin{proof}
The proof easily follows from Theorem~\ref{th:forall_and} and Theorem~\ref{th:equiv:indep}.
By Theorem~\ref{th:forall_and}, $\psi \approx \psi''$ with
$\psi'' \colonequals \psi[\sfrac{\psi_3}{\psi_1}]$ and
$\psi_3 \colonequals \bigl(\forall x':\varphi_1[\sfrac{x'}{x}]\bigr) \land \bigl(\forall x:\varphi_2\bigr)$.
Since $x\notin\var[\varphi_1]$, we have $x'\notin\var[{\varphi_{1}[\sfrac{x'}{x}]}]$, and due to
Theorem~\ref{th:equiv:indep}
$\psi'' \approx \psi'''$ with
$\psi''' \colonequals \psi[\sfrac{\psi_4}{\psi_1}]$ and
$\psi_4 \colonequals \varphi_1[\sfrac{x'}{x}]^{-x'} \land \bigl(\forall x:\varphi_2\bigr)$.
Because of $\varphi_1[\sfrac{x'}{x}]^{-x'} = \varphi_1^{-x}$ we have $\psi_4 = \psi_2$
and thus $\psi \approx \psi'$.
\end{proof}

Finally, in case of rule~\eqref{equiv:exists_or} we need non-trivial additional restrictions
to preserve satisfiability / unsatisfiability of DQBFs where a subformula
$\exists y(D_y):(\varphi_1\lor\varphi_2)$ is replaced by
$\bigl(\exists y(D_y):\varphi_1\bigr) \lor \bigl(\exists y'(D_y):\varphi_2[y'/y]\bigr)$
(or vice versa).

\begin{theorem}
  \label{th:equisat:exists_or}
  Let $\psi \in \npnc$ be a DQBF and let
  $\psi_1 \colonequals \exists y(D_y):(\varphi_1\vee\varphi_2) \sqsubseteq \psi$ be a subformula of $\psi$. Further, let
  \begin{align*}
    \varocc[\varphi_1] &\colonequals \Big[(\varall[\psi] \cap \var[\varphi_1]) \cup \bigcup_{v \in \varex[\psi] \cap \var[\varphi_1]} (\varall[\psi] \cap D_v) \Big] \setminus D_y,\\
    \varocc[\varphi_2] &\colonequals \Big[(\varall[\psi] \cap \var[\varphi_2]) \cup \bigcup_{v \in \varex[\psi] \cap \var[\varphi_2]} (\varall[\psi] \cap D_v)\Big] \setminus D_y,\\
    \varocc[\psi \setminus \psi_1] &\colonequals \Big[\varall[{\psi[\sfrac{0}{\psi_1}]}] \cup \bigcup_{v \in \varex[{\psi[\sfrac{0}{\psi_1}]}]} (\varall[\psi] \cap D_v)\Big]\,.\\
  \end{align*}

  If $\varocc[\varphi_1] \cap \varocc[\varphi_2] = \emptyset$ and
  $\varocc[\varphi_1] \cap \varocc[\psi \setminus \psi_1] = \emptyset$
  (or $\varocc[\varphi_2] \cap \varocc[\psi \setminus \psi_1] = \emptyset$),
  then $\psi\approx\psi'$ where $\psi' \colonequals \psi[\sfrac{\psi_2}{\psi_1}]$ results from $\psi$ by replacing
  $\psi_1$ by
  $\psi_2 \colonequals (\exists y(D_y):\varphi_1)\vee(\exists y'(D_y):\varphi_2[\sfrac{y'}{y}])$
  with $y'$ being a fresh variable.
\end{theorem}

$\varocc[\varphi_1]$ ($\varocc[\varphi_2]$) is the set of all universal variables occurring in $\varphi_1$
($\varphi_2$)
or in dependency sets of existential variables occurring in $\varphi_1$ ($\varphi_2$), reduced by
the dependency set $D_y$ of $y$.
$\varocc[\psi \setminus \psi_1]$ is the set of universal variables occurring in $\psi$ outside the
subformula $\psi_1$ or in dependency sets of existential variables occurring in $\psi$ outside the
subformula $\psi_1$. In the proof we use that $\varocc[\varphi_1] \cap \varocc[\varphi_2] = \emptyset$
implies that --~after replacing existential variables by Skolem functions~--
$\varphi_1$ and $\varphi_2$ do not share universal variables other than those
from $D_y$.

Before we present the proof of Theorem~\ref{th:equisat:exists_or}, we consider two examples
to motivate that it is necessary to add the conditions in the theorem.

\begin{example}[continues=ex:counter]
The first example is again the formula
$\forall x \exists y(\emptyset): \bigl((x \land y) \lor (\neg x \land \neg y)\bigr)$,
which showed that rule~\eqref{equiv:exists_or} cannot be always applied for replacing
subformulas without changing the satisfiability of the formula. We now demonstrate
that one of the conditions from Theorem~\ref{th:equisat:exists_or} does indeed not
hold for this formula.
Using the same notation as in Theorem~\ref{th:equisat:exists_or} we have
$\psi_1 \colonequals \exists y(\emptyset) :\bigl((x \land y) \lor (\neg x \land \neg y)\bigr)$
with $\varphi_1 \colonequals (x \land y)$ and $\varphi_2 \colonequals (\neg x \land \neg y)$.
Then $\varocc[\varphi_1] = \varocc[\varphi_2] = \{x\}$, which means that the condition
$\varocc[\varphi_1] \cap \varocc[\varphi_2] = \emptyset$ is not fulfilled and
Theorem~\ref{th:equisat:exists_or} cannot be applied.
\end{example}

To show that for the correctness of the theorem it is also needed to add condition
$\varocc[\varphi_1] \cap \varocc[\psi \setminus \psi_1] = \emptyset$
(or $\varocc[\varphi_2] \cap \varocc[\psi \setminus \psi_1] = \emptyset$),
we give another example:

\begin{example}[label=ex:counter2]
Let $\psi_1 \colonequals \exists y(\emptyset):\bigl((x_1 \land y) \lor (x_2 \land \neg y)\bigr)$
be a DQBF with $\varphi_1 \colonequals (x_1 \land y)$ and $\varphi_2 \colonequals (x_2 \land \neg y)$.
Let $\psi_2 \colonequals \bigl(\exists y (\emptyset) :(x_1 \land y) \lor \exists y' (\emptyset): (x_2 \land \neg y')\bigr)$
be the DQBF that results from $\psi_1$ by applying rule~\eqref{equiv:exists_or}.
Formula $\psi \colonequals \forall x_1 \forall x_2 : \psi_1 \lor (\neg x_1 \land \neg x_2)$
is then unsatisfiable, because for $s_0(y) = \fzero$ we have $s_0(\psi) \equiv x_2 \lor (\neg x_1 \land \neg x_2)$ and for $s_1(y) = \fone$ we have $s_1(\psi) \equiv x_1 \lor (\neg x_1 \land \neg x_2)$, \ie for both possible choices for the Skolem function candidates we do not obtain a tautology. However, formula $\psi' \colonequals \forall x_1 \forall x_2 : \psi_2 \lor (\neg x_1 \land \neg x_2)$ is satisfiable by $s(y) = \fone$ and $s(y') = \fzero$,
since $s(\psi') \equiv (x_1 \lor x_2) \lor (\neg x_1 \land \neg x_2) \equiv \fone$.

Checking the conditions of Theorem~\ref{th:equisat:exists_or},
we can see that $\varocc[\varphi_1] = \{x_1\}$ and $\varocc[\varphi_2] = \{x_2\}$,
thus $\varocc[\varphi_1] \cap \varocc[\varphi_2] = \emptyset$; however $\varocc[\psi \setminus \psi_1] = \{x_1, x_2\}$,
which means that $\varocc[\varphi_1] \cap \varocc[\psi \setminus \psi_1] \neq \emptyset$
as well as $\varocc[\varphi_2] \cap \varocc[\psi \setminus \psi_1] \neq \emptyset$.
Thus Theorem~\ref{th:equisat:exists_or} cannot be applied.
\end{example}

Now we come to the proof which shows that the conditions in the theorem are sufficient.
After the proof, we will give an illustration of the construction
by considering Example~\ref{ex:counter2} again.

\begin{proof}
  In this proof we assume that the condition $\varocc[\varphi_1] \cap \varocc[\psi \setminus \psi_1] = \emptyset$
  holds. The proof with condition $\varocc[\varphi_2] \cap \varocc[\psi \setminus \psi_1] = \emptyset$ can be
  done with exactly the same arguments.

  To prove the correctness of Theorem~\ref{th:equisat:exists_or} we show that
  $\sem{\psi}\neq\emptyset$ iff $\sem{\psi'}\neq\emptyset$.

  First, assume $\sem{\psi}\neq\emptyset$ and $s\in\sem{\psi}$.
  The function $s'$ with $s'(y')\colonequals s(y)$ and $s'(v) \colonequals s(v)$ otherwise
  is a valid Skolem function for $\psi'$ as well,
  since $s'(\psi') = s(\psi)$.

  \bigskip
  \noindent Now assume that $\sem{\psi'}\neq\emptyset$ and let $s'\in\sem{\psi'}$.
  We construct a Skolem function candidate $s$ of $\psi$ as follows:

  $s(v) \colonequals s'(v)$ for all $v \in \bigl(\varex[\psi] \setminus \{y\}\bigr) \cup \varfree[\psi]$.

  The definition of $s(y)$
  for each $\mu \in \assign(\varall[\psi])$
  is derived from $s'(y)$ and $s'(y')$:
  \[
  s(y)(\mu) \colonequals \begin{cases}
    s'(y)(\mu)  & \mbox{if } \mu'\bigl(s'(\varphi_1)\bigr) = 1 \
                  \forall \mu' \in \assign(\varall[\psi']) \mbox{ with } \mu'|_{D_y} = \mu|_{D_y}; \\
    s'(y')(\mu) & \mbox{otherwise.}
  \end{cases}
  \]
  Note that by this definition $s(y)$ only depends on universal variables from $D_y$, \ie we have defined
  a valid Skolem function candidate according to Definition~\ref{def:skolem_function_candidates}.

  We prove that $s$ is a Skolem function for $\psi$ by contradiction:
  Assume that there exists $\mu \in \assign(\varall[\psi]) = \assign(\varall[\psi'])$ with
  $\mu\bigl(s(\psi)\bigr) = 0$. $\mu\bigl(s'(\psi')\bigr) = 1$, since $s'(\psi')$ is a tautology.
  According to Lemma~\ref{lemma:monotonic}, $\mu\bigl(s'(\psi')\bigr) = 1$ and $\mu\bigl(s(\psi)\bigr) = 0$ implies
  $\mu\bigl(s'(\psi_2)\bigr) = 1$ and $\mu\bigl(s(\psi_1)\bigr) = 0$.
  Remember that we have
  $\psi_1 = \exists y(D_y):(\varphi_1\vee\varphi_2)$ and
  $\psi_2 = (\exists y(D_y):\varphi_1)\vee(\exists y'(D_y):\varphi_2[\sfrac{y'}{y}])$.
  Now we consider $\mu\bigl(s'(\varphi_1)\bigr)$ and $\mu\bigl(s'(\varphi_2[\sfrac{y'}{y}])\bigr)$
  and differentiate between four cases:
\begin{description}
  \item[Case 1:] $\mu\bigl(s'(\varphi_1)\bigr) = 0$, $\mu\bigl(s'(\varphi_2[\sfrac{y'}{y}])\bigr) = 0$. \\
  This would contradict $\mu\bigl(s'(\psi_2)\bigr) = \mu\bigl(s'(\varphi_1)\bigr) \lor \mu\bigl(s'(\varphi_2[\sfrac{y'}{y}])\bigr) = 1$.

  \item[Case 2:] $\mu\bigl(s'(\varphi_1)\bigr) = 1$, $\mu\bigl(s'(\varphi_2[\sfrac{y'}{y}])\bigr) = 1$. \\
  Since we either have $s(y)(\mu) = s'(y)(\mu)$ or $s(y)(\mu) = s'(y')(\mu)$, this
  implies $\mu\bigl(s(\varphi_1)\bigr) = 1$ or $\mu\bigl(s(\varphi_2)\bigr) = 1$. In both cases we obtain
  $\mu\bigl(s(\psi_1)\bigr) = \mu\bigl(s(\varphi_1)\bigr) \lor \mu\bigl(s(\varphi_2)\bigr) = 1$, which contradicts
  $\mu\bigl(s(\psi_1)\bigr) = 0$.

  \item[Case 3:] $\mu\bigl(s'(\varphi_1)\bigr) = 0$, $\mu\bigl(s'(\varphi_2[\sfrac{y'}{y}])\bigr) = 1$. \\
  From the definition of $s(y)$ we obtain $s(y)(\mu) = s'(y')(\mu)$
  and therefore
  $\mu\bigl(s(\varphi_2)\bigr) = \mu\bigl(s'(\varphi_2[\sfrac{y'}{y}])\bigr) = 1$
  and thus $\mu\bigl(s(\psi_1)\bigr) = \mu\bigl(s(\varphi_1)\bigr) \lor \mu\bigl(s(\varphi_2)\bigr) = 1$.
  Again, this contradicts $\mu\bigl(s(\psi_1)\bigr) = 0$.

  \item[Case 4:] $\mu\bigl(s'(\varphi_1)\bigr) = 1$, $\mu\bigl(s'(\varphi_2[\sfrac{y'}{y}])\bigr) = 0$. \\
  Our proof strategy is to show

  \textbf{Fact 1:}  $s(y)(\mu) = s'(y)(\mu)$.

  \noindent Then we would have $\mu\bigl(s(\varphi_1)\bigr) = \mu\bigl(s'(\varphi_1)\bigr) = 1$ and thus
  $\mu\bigl(s(\psi_1)\bigr) = \mu\bigl(s(\varphi_1)\bigr) \lor \mu\bigl(s(\varphi_2)\bigr) = 1$, which would again
  contradict $\mu\bigl(s(\psi_1)\bigr) = 0$.
\end{description}

  So it remains to show Fact 1.
  According to the definition of $s(y)$, $s(y)(\mu) = s'(y)(\mu)$
  iff for all $\mu' \in \assign(\varall[\psi'])$ with $\mu'_{|{D_y}} = \mu_{|{D_y}}$
  we have $\mu'\bigl(s'(\varphi_1)\bigr) = 1$.
  Assume that this is not the case, \ie there is a $\mu' \in \assign(\varall[\psi'])$ with $\mu'|_{D_y} = \mu|_{D_y}$
  and $\mu'\bigl(s'(\varphi_1)\bigr) = 0$.
  Now we show that if this were the case, then we could construct another $\mu'' \in \assign(\varall[\psi'])$
  with $\mu''\bigl(s'(\psi')\bigr) = 0$ which would contradict the fact that $s'(\psi')$ is a tautology.
  Define $\mu'' \in \assign(\varall[\psi'])$ by
  $\mu''(x) \colonequals \mu'(x)$ for all $x \in \varocc[\varphi_1]$, $\mu''(x) \colonequals \mu(x)$ otherwise.
  Since $\mu'|_{D_y} = \mu|_{D_y}$, we have $\mu''|_{D_y} = \mu'|_{D_y} = \mu|_{D_y}$.
  Since $s'(\varphi_1)$ only contains variables from $\varocc[\varphi_1]\cup D_y$ we have
  $\mu''\bigl(s'(\varphi_1)\bigr) = \mu'\bigl(s'(\varphi_1)\bigr) = 0$.
  Moreover, because $\varocc[\varphi_1] \cap \varocc[\varphi_2] = \emptyset$ (from the precondition of the theorem)
  and $s'(\varphi_2[\sfrac{y'}{y}])$ contains only variables from $\varocc[\varphi_2] \cup D_y$, we have $\mu''\bigl(s'(\varphi_2[\sfrac{y'}{y}])\bigr) = \mu\bigl(s'(\varphi_2[\sfrac{y'}{y}])\bigr) = 0$.
  Altogether we have
  $\mu''\bigl(s'(\psi_2)\bigr) = \mu''\bigl(s'(\varphi_1)\bigr) \lor \mu''\bigl(s'(\varphi_2[\sfrac{y'}{y}])\bigr) = 0$
  and therefore $\mu''\bigl(s'(\psi')\bigr) = \mu''\bigl(s'(\psi'[\sfrac{0}{\psi_2}])\bigr)$.
  Since $\mu''$ can differ from $\mu$ only for variables in $\varocc[\varphi_1]$
  and those universal variables do not occur outside $s'(\psi_2)$ due to
  the precondition $\varocc[\varphi_1] \cap \varocc[\psi \setminus \psi_1] = \emptyset$, we
  further have $\mu''\bigl(s'(\psi'[\sfrac{0}{\psi_2}])\bigr) = \mu\bigl(s'(\psi'[\sfrac{0}{\psi_2}])\bigr)$.
  By definition of $s$ as well as $\psi$ and $\psi'$, we have $\mu\bigl(s'(\psi'[\sfrac{0}{\psi_2}])\bigr)
  = \mu\bigl(s(\psi[\sfrac{0}{\psi_1}])\bigr)$ which is $0$ by our initial assumption.
  Taking the last equations together, we have $\mu''\bigl(s'(\psi')\bigr) = 0$ which is our announced contradiction to the fact
  that $s'(\psi')$ is a tautology. This completes the proof of Fact 1.

  \medskip

  In all four cases we were able to derive a contradiction and thus our assumption
  that there exists $\mu \in \assign(\varall[\psi])$ with
  $\mu\bigl(s(\psi)\bigr) = 0$ has to be wrong. $s(\psi)$ is a tautology and $\sem{\psi}\neq\emptyset$.
\end{proof}

Now we illustrate the construction of the proof by demonstrating where the construction fails when the conditions
of Theorem~\ref{th:equisat:exists_or} are not satisfied. To do so, we look into Example~\ref{ex:counter2}
again.

\begin{example}[continues=ex:counter2]
We look again at the DQBF $\psi \colonequals \forall x_1 \forall x_2 : \psi_1 \lor (\neg x_1 \land \neg x_2)$
with $\psi_1 \colonequals \exists y(\emptyset):(\varphi_1  \lor \varphi_2)$,
$\varphi_1 \colonequals (x_1 \land y)$, and $\varphi_2 =(x_2 \land \neg y)$.
$\psi'$ results from $\psi$ by replacing
$\psi_1$ by $\psi_2 \colonequals \bigl(\exists y (\emptyset): (x_1 \land y)\bigr) \lor \bigl(\exists y' (\emptyset): (x_2 \land \neg y')\bigr)$.

We already observed that the conditions of Theorem~\ref{th:equisat:exists_or} are not
satisfied for this DQBF, $\psi$ is not satisfiable, and $\psi'$ is satisfiable.
The Skolem function candidate $s'$ with $s'(y) = \fone$ and $s'(y') = \fzero$ is the only one
that satisfies $\psi'$.
Now we consider where and why the construction of a Skolem function $s$ for $\psi$ (as shown
in the proof) fails.

Since $s'(\varphi_1) \equiv x_1 \neq \fone$, the definition of $s(y)$ in the proof
leads to $s(y) = s'(y') = \fzero$.
In order to prove by contradiction that $s(\psi)$ is a tautology, the proof assumes
an assignment $\mu \in \assign\bigl(\{x_1, x_2\}\bigr)$ with
$\mu\bigl(s(\psi)\bigr) = 0$ and $\mu\bigl(s'(\psi')\bigr) = 1$.
$\bigl(\mu(x_1), \mu(x_2)\bigr) = (0, 0)$ is not possible, since in this case
$\mu\bigl(s(\psi)\bigr) = 1$ as well due to $\mu(\neg x_1 \land \neg x_2) = 1$.
Also for the cases $\bigl(\mu(x_1), \mu(x_2)\bigr) = (1, 1)$ and $\bigl(\mu(x_1), \mu(x_2)\bigr) = (0, 1)$
we obtain contradictions as given in the proof.
The interesting case (where the proof fails) is $\bigl(\mu(x_1), \mu(x_2)\bigr) = (1, 0)$.
This implies
$\mu\bigl(s'(\varphi_1)\bigr) = \mu(x_1) = 1$
and
$\mu\bigl(s'(\varphi_2[\sfrac{y'}{y}])\bigr) = \mu(x_2) = 0$, \ie we are in Case 4,
where we try to prove
Fact 1 (\ie $s(y) = s'(y)$ for the constant Skolem functions in the example)
by contradiction -- which does not work here.
Reduced to our example, Fact 1 does not hold
if there is an assignment
$\mu' \in \assign\bigl(\{x_1, x_2\}\bigr)$ with $\mu'\bigl(s'(\varphi_1)\bigr) = 0$, which just means
$\mu'(x_1) = 0$.
The proof idea is to construct from
$\mu'$
another assignment $\mu'' \in \assign\bigl(\{x_1, x_2\}\bigr)$
with $\mu''\bigl(s'(\psi')\bigr) = 0$ (contradicting the fact that $s'(\psi')$ is a tautology).
$\mu''(x) = \mu'(x)$ for all $x \in \varocc[\varphi_1] = \{x_1\}$,
$\mu''(x) = \mu(x)$ otherwise, \ie
$\mu''(x_1) = 0$ and $\mu''(x_2) = 0$,
leading to
$\mu''\bigl(s'(\psi_2)\bigr) = \mu''(x_1 \lor x_2) = 0$.

The contradiction derived in the proof relies on
the fact that $\mu''$ can differ from $\mu$ only for variables in $\varocc[\varphi_1]$,
which implies by the precondition
$\varocc[\varphi_1] \cap \varocc[\psi \setminus \psi_1] = \emptyset$
that the assignments to universal variables outside $s'(\psi_2)$ are identical
both for $\mu$ and $\mu''$. This is not true in our example where
$\varocc[\varphi_1] \cap \varocc[\psi \setminus \psi_1] = \{x_1\}$ and
$\mu(\neg x_1 \land \neg x_2) = 0$, but
$\mu''(\neg x_1 \land \neg x_2) = 1$.
Thus $\mu''\bigl(s'(\psi')\bigr) = \mu''\bigl((x_1 \lor x_2) \lor (\neg x_1 \land \neg x_2)\bigr)=1$, \ie we do not obtain the contradiction $\mu''\bigl(s'(\psi')\bigr) = 0$.
\end{example}

\subsection{Refuting Propositions~\ref{prop4} and \ref{prop5} from \cite{BalabanovCJ14}}

\noindent A first paper looking into quantifier localization for DQBF was \cite{BalabanovCJ14}.
To this end, they proposed Propositions~\ref{prop4} and \ref{prop5}, which are unfortunately unsound.
We literally repeat Proposition~\ref{prop4} from \cite{BalabanovCJ14}:
\setcounter{proposition}{3}
\begin{proposition}[\cite{BalabanovCJ14}]
  \label{prop4}
  The DQBF
  \begin{equation*}
    \label{eq1}
    \forall \vec{x}\, \exists y_1(S_1) \ldots \exists y_m(S_m) : (\phi_A \vee \phi_B)
  \end{equation*}
  where $\forall \vec{x}$ denotes $\forall x_1 \ldots \forall x_n$, sub-formula $\phi_A$
  (respectively $\phi_B$) refers to variables $X_A \subseteq X$ and $Y_A \subseteq Y$
  (respectively $X_B \subseteq X$ and $Y_B \subseteq Y$), is logically equivalent to
  \begin{align*}
    \label{eq2}
    \forall \vec{x}_c\bigl((\forall \vec{x}_a \exists y_{a_1}(S_{a_1} \cap X_A) \ldots \exists y_{a_p}(S_{a_p} \cap X_A): \phi_A) \vee {} \\
    (\forall \vec{x}_b \exists y_{b_1}(S_{b_1} \cap X_B) \ldots \exists y_{b_q}(S_{b_q} \cap X_B): \phi_B)\bigr),
  \end{align*}
  where variables $\vec{x}_c$ are in $X_A \cap X_B$,
  variables $\vec{x}_a$ are in $X_A \setminus X_B$,
  variables $\vec{x}_b$ are in $X_B \setminus X_A$,
  $y_{a_i} \in Y_A$, and $y_{b_j} \in Y_B$.
\end{proposition}

\begin{lemma}
  \label{lemma:prop4unsound}
  Proposition~\ref{prop4} is unsound.
\end{lemma}

\begin{proof}
Consider the following DQBF
\begin{equation*}
    \label{eq:prop4_cex}
    \psi^1 \colonequals \forall x_1 \forall x_2 \exists y_1(x_2) :
    \bigl(\underbrace{(x_1\equiv x_2)}_{\phi_A} \vee \underbrace{(x_1\not\equiv y_1)}_{\phi_B}\bigr)
\end{equation*}

By \eqref{equiv:exists_and}, $\psi^1$ is equisatisfiable with $\psi$ from Example~\ref{ex1}
and thus satisfiable.
According to Proposition \ref{prop4} we can identify the sets
$X_A  = \{x_1,x_2\}$,  $X_B = \{x_1\}$, $Y_A = \emptyset$, and $Y_B = \{y_1\}$.
and rewrite the formula to
\begin{equation}
  \label{eq:prop4_cex2}
  \psi^2 \colonequals \forall x_1 : \big( ( \forall x_2:(x_1\equiv x_2) ) \vee ( \exists y_1(\emptyset):(x_1\not\equiv y_1) ) \big).
\end{equation}
This formula, in contrast to $\psi^1$, is unsatisfiable because the only Skolem functions candidates
for $y_1$ are $\fzero$ and $\fone$. Both Skolem function candidates do not turn $\psi^2$  into a tautology.
\end{proof}
In the example from the proof, the ``main mistake'' was to replace $D_{y_1} = \{x_2\}$ by $\emptyset$.
If this \emph{were} correct, then the remainder would follow from \eqref{equiv:exists_and} and \eqref{equiv:forall_or}.
\begin{remark}
   Proposition~\ref{prop4} of \cite{BalabanovCJ14} is already unsound when we consider the
   commonly accepted semantics of closed prenex DQBFs as stated in Definition~\ref{def:dqbf_semantics_cp}.
   The proposition claims that $\psi^1$  is equisatisfiable with $\psi^2$ . Additionally, it claims
   that
   \begin{equation*}
     \label{eq:prop4_cex3}
     \psi^3 \colonequals \forall x_1 \forall x_2 \exists y_1(\emptyset) :
     \bigl((x_1\equiv x_2) \vee (x_1\not\equiv y_1)\bigr)
   \end{equation*}
   is equisatisfiable with $\psi^2$. Due to transitivity of equisatisfiability,
   Proposition~\ref{prop5} claims that $\psi^1$ is equisatisfiable with $\psi^3$.
   However, according to the
   semantics in Definition~\ref{def:dqbf_semantics_cp}, $\psi^1$ is satisfiable and $\psi^3$
   unsatisfiable. Also note that $\psi^1$ and $\psi^3$ are actually QBFs; so Proposition~\ref{prop4}
   is also unsound when restricted to QBFs.
\end{remark}

Next we literally repeat Proposition~\ref{prop5} from \cite{BalabanovCJ14}:
\setcounter{proposition}{4}
\begin{proposition}[\cite{BalabanovCJ14}]
  \label{prop5}
  The DQBF
  \begin{equation*}
    \label{eq3}
    \forall \vec{x}\,\exists y_1(S_1) \ldots \exists y_k(S_k): (\phi_A \land \phi_B)
  \end{equation*}
  where $\forall \vec{x}$ denotes $\forall x_1 \ldots \forall x_n$, sub-formula $\phi_A$ (respectively $\phi_B$)
  refers to variables $X_A \subseteq X$ and $Y_A \subseteq Y$ (respectively $X_B \subseteq X$ and $Y_B \subseteq Y$),
  is logically equivalent to
  \begin{equation*}
    \label{eq4}
    \forall \vec{x}\,\exists y_2(S_2) \ldots \exists y_k(S_k): \bigl((\exists y_1(S_1 \cap X_A): \phi_A) \land \phi_B\bigr)
  \end{equation*}
  for $y_1 \notin Y_B$.
\end{proposition}

\begin{lemma}
  \label{lemma:prop5unsound}
  Proposition~\ref{prop5} is unsound.
\end{lemma}

\begin{proof}
For a counterexample, consider the formula
\begin{equation*}
  \label{eq:prop5_cex}
  \psi^4 \colonequals
  \forall x_1 \forall x_2 \exists y_1(x_1,x_2)\exists y_2(x_1,x_2) :
   \underbrace{\bigl(y_1\equiv\neg y_2\bigr)}_{\phi_A}
  \land
  \underbrace{\bigl(y_2\equiv (x_1\land x_2)\bigr)}_{\phi_B}.
\end{equation*}
with the corresponding variable sets
$X_A = \emptyset$, $X_B = \{x_1,x_2\}$, $Y_A  = \{y_1,y_2\}$, and $Y_B = \{y_2\}$.
We have $y_1\notin Y_B$ and $\{x_1,x_2\}\cap X_A = \emptyset$.
Proposition~\ref{prop5} says that $\psi^4$ is equisatisfiable with:
\begin{equation*}
  \psi^5 \colonequals
  \forall x_1\forall x_2\exists y_2(x_1,x_2):
    \bigl( \exists y_1(\emptyset): (y_1\equiv\neg y_2)\bigr)\land \bigl(y_2\equiv (x_1\land x_2)\bigr)\,.
\end{equation*}

The formula $\psi^4$ is satisfiable; the Skolem function $s$ with
$s(y_1)=\neg (x_1\land x_2)$ and $s(y_2)=(x_1\land x_2)$ is in $\sem{\psi}$.

The formula $\psi^5$, however, is unsatisfiable: Since $D_{y_1}^{\psi^5} = \emptyset$, there are only two
Skolem function candidates for $y_1$, either $s(y_1)=\fzero$ or $s(y_1)=\fone$. In the first case,
we need to find a function for $y_2$ such that $\bigl(\fzero \equiv\neg y_2\bigr)\land \bigl(y_2\equiv (x_1\land x_2)\bigr)$
becomes a tautology. In order to satisfy the first part, $\fzero \equiv\neg y_2$, we need to set $s(y_2)=\fone$.
Then the formula can be simplified to $(x_1\land x_2)$, which is not a tautology.
In the second case, $s(y_1)=\fone$, we get the expression $\bigl(\fone\equiv\neg y_2\bigr)\land \bigl(y_2\equiv (x_1\land x_2)\bigr)$.
This requires to set $s(y_2)=\fzero$ in order to satisfy the first part, turning the formula into
$\fzero\equiv (x_1\land x_2)$, or more concisely, $\neg(x_1\land x_2)$, which is neither a tautology.
Therefore we can conclude that $\psi^5$ is unsatisfiable and, accordingly, Proposition~\ref{prop5}
of \cite{BalabanovCJ14} is unsound.
\end{proof}

For Proposition~\ref{prop5} we make a similar observation as for Proposition~\ref{prop4}:
\begin{remark}
   Also Proposition~\ref{prop5} of \cite{BalabanovCJ14} is already unsound when we consider the
   commonly accepted semantics of closed prenex DQBFs as stated in Definition~\ref{def:dqbf_semantics_cp}.
   The proposition claims that $\psi^4$ is equisatisfiable with $\psi^5$. Additionally, it claims
   that
   \begin{equation*}
     \psi^6 \colonequals
     \forall x_1\forall x_2\exists y_1(\emptyset)\exists y_2(x_1,x_2):
     \bigl(y_1\equiv\neg y_2\bigr)\land \bigl(y_2\equiv(x_1\land x_2)\bigr)
   \end{equation*}
   is equisatisfiable with $\psi^5$. Due to transitivity of equisatisfiability,
   Proposition~\ref{prop5} claims that $\psi^4 \approx \psi^6$ holds. However, according to the
   semantics in Definition~\ref{def:dqbf_semantics_cp}, $\psi^4$ is satisfiable and $\psi^6$
   unsatisfiable. Again, $\psi^4$ and $\psi^6$ are actually QBFs; so Proposition~\ref{prop5}
   is also unsound when restricted to QBFs.
\end{remark}

%% file: 04-algorithm.tex
\section{Taking Advantage of Quantifier Localization}
\label{sec:algorithm}

\noindent In this section, we explain the implementation of the algorithm that exploits the
properties of non-prenex DQBFs to simplify a given formula.
First, we define necessary concepts and give a coarse sketch of the algorithm.
Then, step by step, we dive into the details.

Benedetti introduced in \cite{Benedetti05c} \textit{quantifier trees} for pushing quantifiers into a CNF.
In a similar way we construct a \textit{quantifier graph}, which is a structure similar to an
And-Inverter Graph (AIG)~\cite{Kuehlmann2002}.
It allows to perform quantifier localization according to Theorems~\ref{th:rules}, \ref{th:forall_and},
\ref{th:equisat:forall_and}, and \ref{th:equisat:exists_or}.

\begin{definition} [Quantifier graph]
  \label{def:quantGraph}
  For a non-prenex DQBF $\psinp$, a \emph{quantifier graph} is a directed acyclic graph $G_{\psinp} = (N_{\psinp},E_{\psinp})$.
  $N_{\psinp}$ denotes the set of nodes of $G_{\psinp}$. Each node $n\in N_{\psinp}$ is labeled with an operation
  ${\op}\in \{\wedge,\vee\}$ from $\psinp$ if $n$ is an inner node, or with a variable $v\in\var[\psinp]$ if it is a terminal node.
  $E_{\psinp}$ is a set of edges. Each edge is possibly augmented with quantified variables and\,/\,or negations.
\end{definition}

The input to the basic algorithm for quantifier localization (\textit{DQBFQuantLoc}),
which is shown in Algorithm~\ref{algo:DQBFQuLo}, is a closed prenex DQBF $\psi$.
The matrix $\varphi$ of $\psi$ is represented as an AIG, and the prefix $Q$ is a set of quantifiers
as stated in Definition~\ref{def:dqbf_cp}.
(If the matrix is initially given in CNF, we preprocess it by circuit extraction
(see for instance \cite{PigorschS09,wimmer-et-al-sat-2015,wimmer-et-al-jsat-2019}) and
the resulting circuit is then represented by an AIG.)
The output of \textit{DQBFQuantLoc} is a DQBF in closed prenex form again.
In intermediate steps, we convert $\psi$ into a non-prenex DQBF $\psinp$,
represented as a quantifier graph, by pushing quantifiers of the prefix into the matrix.
After pushing the quantified variables as deep as possible into the formula, we eliminate quantifiers wherever it is possible.
If a quantifier cannot be eliminated, it is pulled out of the formula again. In this manner we finally obtain a modified and possibly simplified prenex DQBF $\psi'$.

\begin{algorithm}[tb]
\SetKwInOut{Input}{Input}
\SetKwInOut{Output}{Output}
\SetKwFunction{NormalizeToNNF}{NormalizeToNNF}
\SetKwFunction{BuildMacroGates}{BuildMacroGates}
\SetKwFunction{EliminatePures}{EliminatePures}
\SetKwFunction{Localize}{Localize}
\SetKwFunction{Eliminate}{Eliminate}

\Input{Prenex DQBF $\psi\colonequals Q : \varphi$, where\\
	-- $Q = \forall x_1\ldots\forall x_n\exists y_1(D_{y_1})\ldots\exists y_m(D_{y_m})$\\
	-- $\varphi$ has an arbitrary structure given as an AIG}
\Output{prenex DQBF $\psi'$}
\BlankLine

$G_{\psinp} \colonequals$ \NormalizeToNNF{$\psi$}\;\label{algo:qloc:norm}
$G_{\psinp} \colonequals$ \BuildMacroGates{$G_{\psinp}$}\;\label{algo:qloc:macro}
$G_{\psinp} \colonequals$ \Localize{$G_{\psinp}$}\;\label{algo:qloc:loc}
$G_{\psi'}  \colonequals$ \Eliminate{$G_{\psinp}$}\;\label{algo:qloc:elim}
\Return {$\psi'$}
\caption{DQBFQuantLoc} \label{algo:DQBFQuLo}
\end{algorithm}

\subsection{Building NNF and Macrogates}
\label{ssec:create_nnf}

\begin{figure}[tb]
\centering
\scalebox{.9}{\input{psi_NNF}}
\caption{Quantifier graph in NNF.} \label{ex:quantgraph}
\end{figure}
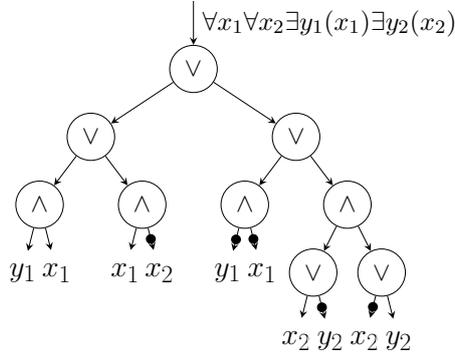

\noindent In Line~\ref{algo:qloc:norm} of Algorithm~\ref{algo:DQBFQuLo},
we first translate the matrix $\varphi$ of the DQBF $\psi$ into negation normal
form (NNF) by pushing the negations of the circuit to the primary inputs
(using De Morgan's law).
The resulting formula containing the matrix in NNF is represented as a
quantifier graph as in Definition~\ref{def:quantGraph}, where we only
have negations at those edges which point to terminals.
Figure~\ref{ex:quantgraph} shows a quantifier graph as returned by
\textit{NormalizeToNNF}. We will use it as a running example to illustrate
our algorithm.

Then, in Line~\ref{algo:qloc:macro} of Algorithm~\ref{algo:DQBFQuLo},
we combine subcircuits into \emph{AND}\,/\,\emph{OR} \emph{macrogates}.
The combination into macrogates is essential to increase the degrees of freedom given
by different decompositions of \emph{AND}s\,/\,\emph{OR}s that enable different applications
of the transformation rules according to rules~\eqref{equiv:forall_or}, \eqref{equiv:exists_and}, and
Theorems~\ref{th:forall_and}, \ref{th:equisat:forall_and}, and \ref{th:equisat:exists_or}.
A \emph{macrogate} is a multi-input gate, which we construct by collecting consecutive nodes representing
the same logic operation ($\land$, $\lor$).
Except for the topmost node within a macrogate no other node may have more than one incoming edge,
\ie macrogates are subtrees of fanout-free cones.
During the collection of nodes, we stop the search along a path when we visit a node with multiple parents.
From this node we later start a new search. The nodes which are the target of an edge leaving a
macrogate are the \emph{macrochildren} of the macrogate and the parents of its root are called the \emph{macroparents}.
It is clear that a macrogate consisting of only one node has exactly two children like a standard node.
For such nodes we use the terms macrogate and node interchangeably.
In Figure~\ref{ex:macrogate} we show a macrogate found in the running example.

\subsection{Quantifier Localization}
\label{ssec:localize}

\noindent After calling \textit{NormalizeToNNF} and \textit{BuildMacroGates}
the only edge that carries quantified variables is the root edge.
By shifting quantified variables to edges below the root node we push them into the formula. Sometimes we say that we push a quantified variable to a child by which we mean that we write the variable to the edge pointing to this child.

On the quantifier graph for formula $\psinp$ we perform the localization of quantifiers
according to Theorems~\ref{th:rules},
\ref{th:forall_and}, \ref{th:equisat:forall_and}, and \ref{th:equisat:exists_or}
with the function \textit{Localize} in Line~\ref{algo:qloc:loc} of Algorithm~\ref{algo:DQBFQuLo}.
Algorithm~\ref{algo:local} presents the details.

\begin{algorithm}[tb]
	\SetKwInOut{Input}{Input}
	\SetKwInOut{Output}{Output}
	\SetKwFunction{PushVariables}{PushVariables}
	\SetKwFunction{SeparateIncomings}{SeparateIncomings}
	
	\Input{Quantifier graph $G_{\psinp}$}
	\Output{Modified quantifier graph $G_{\psinp}$}
	\BlankLine
	nodeList $\colonequals$ double linked list of all nodes of $G_{\psinp}$ in
	   topological order, \ie starting with the root gate\;
	$g\colonequals$ nodeList.first \;
	\While{$g \neq \mathrm{nodeList.end}$}{\label{algo:local:while}
		\PushVariables{g}\;\label{algo:local:push}
		\uIf{there are variables left which do not occur on all incoming edges of $g$}{
			$v\colonequals$ variable that occurs on most incoming edges of $g$\;\label{algo:local:pickvar}
			$g'\colonequals$ \SeparateIncomings{g,v}\;\label{algo:local:separate}
			insert $g'$ into nodeList after $g$\;
		}\Else{
            $g\colonequals$ nodeList.next\;
		}
	}
	\Return modified quantifier graph $G_{\psinp}$\;\label{algo:local:end}
	\caption{Localize} \label{algo:local}
\end{algorithm}

The graph is traversed in topological order (given by nodeList),
starting with the macrogate $g_{root}$,
which is the root of the graph representing $\psinp$.
Note that despite the transformations made so far, the graph passed to \textit{Localize} still represents a prenex DQBF.
For each macrogate $g$ from the graph we first call the function \textit{PushVariables} in Algorithm~\ref{algo:local}, Line~\ref{algo:local:push} (the details of \textit{PushVariables} are listed in Algorithm~\ref{algo:push}).
This function pushes as many variables as possible over $g$.

Only variables that are common to \emph{all} incoming edges of $g$ can be pushed over $g$.
Thus, after \textit{PushVariables} there might be remaining variables, which only occur on some but not all edges.
In Line~\ref{algo:local:pickvar} of Algorithm~\ref{algo:local} we pick such a variable $v$ that occurs on at least one, but not all incoming edges.
To allow $v$ to be pushed over $g$, we apply \textit{SeparateIncomings} to $g$ \wrt $v$ (Algorithm~\ref{algo:local}, Line~\ref{algo:local:separate}).
This function creates a copy $g'$ of $g$ and removes from $g'$ all incoming edges containing $v$. From $g$ it removes all incoming edges without $v$, \ie $v$ occurs on all incoming edges to $g$ when returning from this function. Then, the procedure that pushes variables and possibly copies gates is repeated for $g$ and $g'$.
If there is no more quantified variable for which the incoming edges need to be separated, we continue the procedure 
for the next macrogates in topological order given by nodeList.

Now we take a closer look at the function \textit{PushVariables} from Line~\ref{algo:local:push} of Algorithm~\ref{algo:local}, which pushes quantified variables over a single macrogate $g$ according to Theorems~\ref{th:rules}, \ref{th:forall_and}, \ref{th:equisat:forall_and}, and \ref{th:equisat:exists_or}.

\begin{algorithm}[tb] 
\SetKwInOut{Input}{Input}
\SetKwInOut{Output}{Output}
\SetKwFunction{BuildMacroGates}{BuildMacroGates}
\SetKwFunction{CollectCommonVariables}{CollectCommonVariables}
\SetKwFunction{FindBestVariableConj}{FindBestVariableConj}
\SetKwFunction{FindBestVariableDisj}{FindBestVariableDisj}
\SetKwFunction{IsVarPushable}{IsVarPushable}

\Input{macrogate $g$}
\BlankLine

$V_{\text{com}}\colonequals$ \CollectCommonVariables{g}\;

\uIf{g is a conjunction}{\label{algo:push:conj}
	\While{$V_{\text{com}}\cap\varex\neq\emptyset$}{
		$v\colonequals$ \FindBestVariableConj{$V_{\text{com}}$}\;\label{algo:push:bestconj}
		try to push $v$ to macrochildren\;
		delete $v$ from $V_{\text{com}}$\;
	}
	\For{each universal variable $x$ in $V_{\text{com}}$}{
		try to push $x$ to macrochildren\;
		delete $x$ from $V_{\text{com}}$\;
	}
}
\Else{\label{algo:push:disj}
	\For{each existential variable $y$ in $V_{\text{com}}$}{ \label{algo:push:foreachex}
		\If{\IsVarPushable{y,g}}{\label{algo:push:ispushable}
			push $y$ to macrochildren\;
			delete $y$ from $V_{\text{com}}$\;
		}
	}\label{algo:push:foreachexend}

	\While{$V_{\text{com}}\neq\emptyset$}{\label{algo:push:remaining}
		$v\colonequals$ \FindBestVariableDisj{$V_{\text{com}}$}\;\label{algo:push:bestdisj}
		try to push $v$ to macrochildren\;
		delete $v$ from $V_{\text{com}}$\;
	}\label{algo:push:remainingend}
}
\caption{PushVariables} \label{algo:push}
\end{algorithm}

For macrogate $g$ we first determine the set of quantified variables $V_{\text{com}}$ that occur on \emph{all} incoming edges of $g$ by \textit{CollectCommonVariables}.
These are the only ones which we can push further into the graph.

We annotate the edges by sets of quantifiers, not by sequences of quantifiers, although in
the DQBF formulas we have sequences of quantifiers. The reason for this approach lies
in rules \eqref{equiv:exists_exists}, \eqref{equiv:forall_forall}, and \eqref{equiv:forall_exists}.
The order of quantifiers of the same type can be changed arbitrarily due to rules \eqref{equiv:exists_exists}
and \eqref{equiv:forall_forall}. For quantifiers $\forall x$ and $\exists y(D_y)$ there are only
two possible cases: If $x\in D_y$, then $\forall x$ has to be to the left of $\exists y(D_y)$,
since $\exists y(D_y)\,\forall x:\varphi$ with $x\in D_y$ cannot occur in a valid DQBF by
construction. If $x \notin D_y$, then both orders of $\forall x$ and $\exists y(D_y)$
are possible due to rule~\eqref{equiv:forall_exists}. Thus, by rule~\eqref{equiv:forall_exists}
we can easily derive from a set of quantifiers all orders that are allowed in a valid DQBF.
Since in both cases ($x\in D_y$ and $x \notin D_y$), $\exists y(D_y)$ can be on the right,
it is always possible to push existential quantifiers first.
Universal variables $x$ can only be pushed, if all existential variables $y$ on the same edge do not
contain $x$ in their dependency set (see \eqref{equiv:forall_exists}).
So pushing existential variables first may have the advantage that this enables pushing
of universal variables.

In Lines~\ref{algo:push:conj} and \ref{algo:push:disj} of Algorithm~\ref{algo:push} we distinguish between a conjunction and a disjunction.

\subsubsection{Pushing over Conjunctions}
\label{ssec:push_conj}

\noindent As already said, we push existential variables first.
If $g$ is a conjunction, we can push existential variables using rule \eqref{equiv:exists_and}
only.
Before pushing an existential variable $y$, we collect the set $C_y$ of all children containing $y$.
(Note that we do not differentiate here between a child $c_i$ and the subformula represented by $c_i$,
which would be more precise.)

If $C_y = \emptyset$, then we can just remove the existential quantification of $y$ from the edges towards
$g$ by Theorem~\ref{th:equiv:indep_exists}.
If $C_y = \{c_i\}$, then we simply push the existential variable to the edge leading from $g$ to $c_i$
($c_i$ can then be regarded as $\varphi_2$ from \eqref{equiv:exists_and}).
If $C_y$ contains all children of $g$, then $y$ cannot be pushed.
In all other cases a decomposition of the macrogate $g$ takes place.
All children from $C_y$ are merged and treated as $\varphi_2$ from \eqref{equiv:exists_and},
\ie the \emph{AND} macrogate is decomposed into
one \emph{AND} macrogate $g'$ combining the children in $C_y$, and another
macrogate $g''$ (replacing $g$) whose children are the
remaining children of $g$ as well as the new $g'$.
Pushing $y$, we write $y$ on the incoming edge of $g'$.
(Of course, $g'$ has to be inserted after $g$ into the topological order nodeList used in Algorithm~\ref{algo:local}.)
According to rule \eqref{equiv:exists_exists} we can push existential
variables in an arbitrary order. Here we apply \textit{FindBestVariableConj} (Line~\ref{algo:push:bestconj}, Algorithm~\ref{algo:local})
to heuristically determine a good order of pushing.
We choose the existential variable $y$ first that occurs in the fewest children of $g$,
\ie we choose the variable $y$ where $|C_y|$ is minimal. Remember that the children in $C_y$ are combined into an
\emph{AND} macrogate $g'$ and $y$ (as well as all universal variables $y$ depends on)
cannot be pushed over $g'$. So our goal is to find a variable
$y$ which is the least obstructive for pushing other variables.

Subsequently, only universal variables are left for pushing. This is done by
Theorems~\ref{th:forall_and} and \ref{th:equisat:forall_and}.
As mentioned above, a universal variable $x$ cannot be pushed if there is some existential
variable $y$ with $x \in D_y$ left on the incoming edges of $g$, because $y$ could not be pushed before.
For all other universal variables $x$ we collect the set $C_x$ of all children containing $x$.
If $C_x = \emptyset$, then we just remove the universal quantification of $x$ by Theorem~\ref{th:equiv:indep}.
Otherwise, for all children $c \notin C_x$ we remove $x$ from the dependency sets of all existential
variables $y$ on the edge from $g$ to $c$ according to Theorem~\ref{th:equisat:forall_and}.
For all children $c \in C_x$ we push $x$ to the edge from $g$ to $c$
together with renaming $x$ into a fresh variable $x'$ according to Theorem~\ref{th:forall_and}.

\subsubsection{Pushing over Disjunctions}
\label{ssec:push_disj}

\noindent In case $g$ is a disjunction, at first we check for each existential variable $y$
(Lines~\ref{algo:push:foreachex}--\ref{algo:push:foreachexend}, Algorithm~\ref{algo:push})
whether it can be distributed to its children. This is not always the case, since
the preconditions of Theorem~\ref{th:equisat:exists_or} may possibly not be fulfilled.
Function \textit{IsVarPushable} from Line~\ref{algo:push:ispushable} in Algorithm~\ref{algo:push}
performs this check based on Theorem~\ref{th:equisat:exists_or} and / or rule~\eqref{equiv:exists_and}.
If the function returns \textit{true}, $y$ can be distributed to certain children of $g$.
For the check of function \textit{IsVarPushable}, whose details are listed in Algorithm~\ref{algo:isvarpushable}, we collect
for each existential variable $y$ the set $C_y$ of children containing $y$.
For the check in \textit{IsVarPushable} we do not have to consider the children which are not in $C_y$,
since we do not need to push $y$ to those children according to rule~\eqref{equiv:exists_and}.
If $C_y = \emptyset$, we can just remove the existential quantification of $y$ from the edges towards
$g$ by Theorem~\ref{th:equiv:indep_exists}.
If $C_y = \{c_i\}$, then we can push $y$ to the edge leading to $c_i$ according to rule~\eqref{equiv:exists_and}.
In both cases \textit{IsVarPushable} returns true, see Lines~\ref{algo:isvarpushable:trivial1}--\ref{algo:isvarpushable:trivial2}
of Algorithm~\ref{algo:isvarpushable}.
The remaining cases are handled by Theorem~\ref{th:equisat:exists_or}.
Let $\psinp$ be the DQBF represented by the root node of the quantifier graph.
For each $c_i \in C_y$ we compute as in Theorem~\ref{th:equisat:exists_or}
$\varocc[c_i] = \Big[(\varall[\psinp] \cap \var[c_i]) \cup \bigcup_{v \in \varex[\psinp] \cap \var[c_i]} (\varall[\psinp]\cap D_v)\Big] \setminus D_y$,
the set of all universal variables occurring in $c_i$
or in dependency sets of existential variables occurring in $c_i$, reduced by
the dependency set $D_y$ of $y$ (Lines~\ref{algo:isvarpushable:vocc}--\ref{algo:isvarpushable:voccend}, Algorithm~\ref{algo:isvarpushable}).
Moreover we compute in Line~\ref{algo:isvarpushable:collectoutside} of Algorithm~\ref{algo:isvarpushable}
$\varocc[\psinp \setminus g] = \Big[
\varall[{\psinp[\sfrac{0}{g}]}] \cup
\bigcup_{v \in \varex[{\psinp[\sfrac{0}{g}]}]} (\varall[\psinp]\cap D_v)\Big]$,
the set of universal variables outside the
subformula (represented by the macrogate) $g$ or in dependency sets of existential variables outside the
subformula $g$.
If all $\varocc[c_i]$ are pairwise disjoint and $\varocc[c_i] \cap \varocc[\psinp \setminus g]=\emptyset$ for all
$c_i \in C_y$ except at most one (which then ``plays the role of $\varphi_2$ in Theorem~\ref{th:equisat:exists_or}''),
then \textit{IsVarPushable} returns \textit{true} and
$y$ can be pushed to all $c_i \in C_y$. Again, according to Theorem~\ref{th:equisat:exists_or},
we have to replace $y$ by fresh variables after pushing to the children $c_i \in C_y$.
As already mentioned, we do not need to push $y$ to the children $c_i \notin C_y$ due to
rule~\eqref{equiv:exists_and}.

After checking Theorem~\ref{th:equisat:exists_or} for all existential variables, \textit{PushVariables} continues with all variables left in $V_{com}$ in Lines~\ref{algo:push:remaining}--\ref{algo:push:remainingend} of Algorithm~\ref{algo:push}.
They can be universal variables and those existential variables which have been determined to be stuck due to \textit{IsVarPushable}.
An existential variable $y$ can still be pushed according to \eqref{equiv:exists_and} by
decomposing $g$ into a macrogate $g'$ combining all children in $C_y$ and a macrogate $g''$
combining all other children together with $g'$, with $g''$ replacing $g$ (as in the case of $g$ being a conjunction).
For universal variables $x$ we compute the set $C_x$ of all children containing $x$ or an existential variable $y$ with $x \in D_y$, see rule~\eqref{equiv:forall_or}.
Then we decompose $g$ with a new macrogate combining all children in $C_x$.
Similar to the case of conjunctions, we determine the next variable to be considered for
splitting by \textit{FindBestVariableDisj} (Line~\ref{algo:push:bestdisj}, Algorithm~\ref{algo:push}).
\textit{FindBestVariableDisj} selects the variable that has the fewest children in
$C_y$ resp.~$C_x$ to be split off. For universal variables $x$ we have to take additionally into account
that $x$ is only eligible by \textit{FindBestVariableDisj}, if $x$ is not in the dependency set
of any $y$ on an incoming edge of $g$ anymore, see rule~\eqref{equiv:forall_exists}.

\begin{algorithm}[tb] 
\SetKwInOut{Input}{Input}
\SetKwInOut{Output}{Output}

\Input{Existential variable $y$, disjunctive macrogate $g$}
\Output{true/false}
\BlankLine

$C_y\colonequals$ all children $c_i$ of $g$ with $y\in\var[c_i]$\;\label{algo:isvarpushable:collect}

\If{$|C_y| \leq 1$}{\label{algo:isvarpushable:trivial1}
	\Return true\;\label{algo:isvarpushable:trivial2}
}

\For{each child $c_i$}{\label{algo:isvarpushable:vocc}
	$V^{\forall,occ}_{c_i}\colonequals\Big[(\varall[\psinp]\cap V_{c_i})\cup\bigcup_{v\in\varex[\psinp]\cap V_{c_i}}(\varall[\psinp]\cap D_v)\Big]\setminus D_y$\;
}\label{algo:isvarpushable:voccend}

\For{each pair $c_i,c_j\in C_y,$ $i\neq j$}{\label{algo:isvarpushable:childrendisj}
	\If{$V^{\forall,occ}_{c_i}\cap V^{\forall,occ}_{c_j}\neq\emptyset$}{
		\Return false\;
	}
}

$V^{\forall,occ}_{\psinp\setminus g}\colonequals
\Big[
\varall[{\psinp[\sfrac{0}{g}]}]\cup\bigcup_{v\in\varex[{\psinp[\sfrac{0}{g}]}]}(\varall[\psinp]\cap D_v)
\Big]$\;\label{algo:isvarpushable:collectoutside}

\BlankLine
nonDisjoint $\colonequals$ 0\;

\For{each child $c_i\in C_y$}{\label{algo:isvarpushable:outsidedisj}
	\If{$V^{\forall,occ}_{c_i}\cap V^{\forall,occ}_{\psinp\setminus g}\neq\emptyset$}{
        nonDisjoint $\colonequals$ nonDisjoint + 1\;
        \uIf{nonDisjoint $> 1$}{
			\Return false\;		
 		}     
	}
}
\Return true\;

\caption{IsVarPushable} \label{algo:isvarpushable}
\end{algorithm}

The complete procedure is illustrated in Figure~\ref{ex:localize}.

\begin{figure}[tb]
\centering
\begin{subfigure}{.44\textwidth}
        \centering
	\scalebox{.9}{\input{psi_macro}}
	\caption{A macrogate, marked in red.}
	\label{ex:macrogate}
\end{subfigure}\hfill
\begin{subfigure}{.46\textwidth}
	\centering
	\scalebox{.9}{\input{psi_pushEx}}
	\caption{$G_{\psinp}$ after distributing $y_1$, $y_2$ according\\ to Theorem~\ref{th:equisat:exists_or} and \eqref{equiv:exists_and}.}
	\label{ex:localize_distr}
\end{subfigure}

\begin{subfigure}{.44\textwidth}
	\centering
	\scalebox{.9}{\input{psi_pushAll}}
	\caption{$G_{\psinp}$ after splitting the macrogate to \\ enable pushing $x_2$.}
	\label{ex:localize_restruct}
\end{subfigure}\hfill
\begin{subfigure}{.46\textwidth}
	\centering
	\scalebox{.9}{\input{psi_final}}
	\caption{$G_{\psinp}$ after processing all macrogates.}
	\label{ex:localize_final}
\end{subfigure}
\caption{\textit{BuildMacroGates} and \textit{Localize}.} \label{ex:localize}
\end{figure}
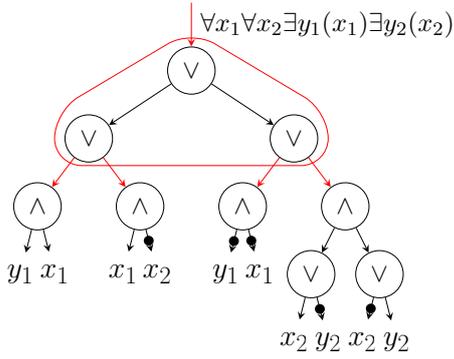
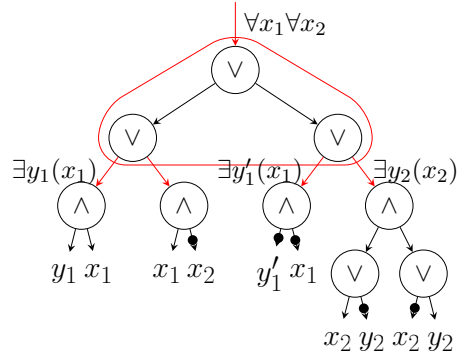
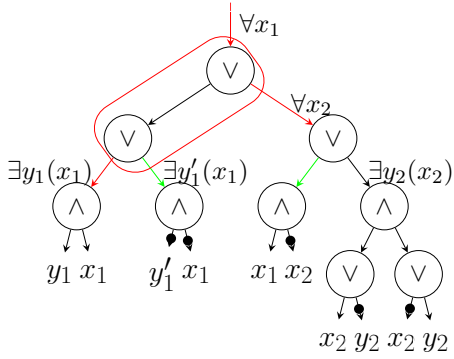
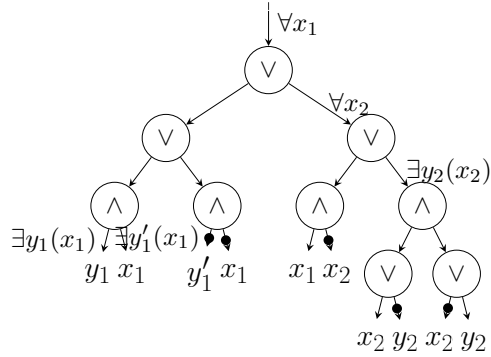

\subsection{Eliminating Variables}
\label{ssec:elim}

\noindent Finally, in Line~\ref{algo:qloc:elim} of Algorithm~\ref{algo:DQBFQuLo}, we try to eliminate those variables which can be
symbolically quantified after quantifier localization.
The conditions are given by Theorem~\ref{th:rules2} and rule~\eqref{th:rules3}. We proceed from the terminals to the root and check each edge with at least one quantified variable written to it.
If a variable could not be eliminated, we pull it back to the incoming edges of this edges' source node. If a variable has been duplicated according to Theorems~\ref{th:forall_and} or \ref{th:equisat:exists_or} and some duplications are brought back to one edge, then we merge them into a single variable again.

As Figure~\ref{ex:localize_final} shows, we can eliminate $y_1$ and $y_1'$,
since there are no other variables in the support of the target nodes.
The same holds for $y_2$ because $x_2$ is the only variable different from $y_2$ in the support of the target node and $x_2$ is in the dependency set of $y_2$, see Theorem~\ref{th:rules2}.
Subsequently, $x_2$ and $x_1$ can be eliminated according to rule~\eqref{th:rules3},
such that we obtain a constant function.

Note that in our implementation we avoid copying and renaming
variables when we apply Theorems~\ref{th:forall_and} or \ref{th:equisat:exists_or}.
This saves additional effort when pulling back variables which could not be eliminated and
avoids to copy shared subgraphs which become different by renaming.
On the other hand, we have to consider this implementation detail in the check of conditions of Theorem~\ref{th:equisat:exists_or}, of course.
Sets $\varocc[c_i]$ and $\varocc[c_j]$ which contain only common
universal variables $x$ which are ``virtually different'' (\ie different, if we would perform renaming)
are then considered to be disjoint. Fortunately, it is easy to decide this by checking whether
$c_i$ and $c_j$ are in the scope of $\forall x$ or not.

Finally, having all remaining variables pulled back to the root edge, we return to a closed prenex DQBF with potentially fewer variables, fewer dependencies and a modified matrix, which we can pass back to
a solver for prenex DQBFs.

%% file: psi_NNF.tex
\begin{tikzpicture}[>=stealth]

\node[draw, shape=circle](1) at (2.75,4){$\vee$};
\node[draw, shape=circle](2) at (1.25,3){$\vee$};
\node[draw, shape=circle](3) at (4.25,3){$\vee$};
\node[draw, shape=circle](4) at (.5,2){$\wedge$};
\node[draw, shape=circle](5) at (2,2){$\wedge$};
\node[draw, shape=circle](6) at (3.5,2){$\wedge$};
\node[draw, shape=circle](7) at (5,2){$\wedge$};
\node[draw, shape=circle](8) at (4.5,1){$\vee$};
\node[draw, shape=circle](9) at (5.5,1){$\vee$};

\node[draw, shape=circle,white,scale=.5](10) at (0.25,1){\textcolor{black}{\Huge $y_1$}};
\node[draw, shape=circle,white,scale=.5](11) at (.75,1){\textcolor{black}{\Huge $x_1$}};
\node[draw, shape=circle,white,scale=.5](12) at (1.75,1){\textcolor{black}{\Huge $x_1$}};
\node[draw, shape=circle,white,scale=.5](13) at (2.25,1){\textcolor{black}{\Huge $x_2$}};
\node[draw, shape=circle,white,scale=.5](14) at (3.25,1){\textcolor{black}{\Huge $y_1$}};
\node[draw, shape=circle,white,scale=.5](15) at (3.75,1){\textcolor{black}{\Huge $x_1$}};
\node[draw, shape=circle,white,scale=.5](16) at (4.25,0){\textcolor{black}{\Huge $x_2$}};
\node[draw, shape=circle,white,scale=.5](17) at (4.75,0){\textcolor{black}{\Huge $y_2$}};
\node[draw, shape=circle,white,scale=.5](18) at (5.25,0){\textcolor{black}{\Huge $x_2$}};
\node[draw, shape=circle,white,scale=.5](19) at (5.75,0){\textcolor{black}{\Huge $y_2$}};

\draw[->](2.75,5) -- (1)node[pos=0.5,right]{ $\forall x_1 \forall x_2 \exists y_1 (x_1) \exists y_2 (x_2)$};
\draw[->](1) -- (2);
\draw[->](1) -- (3);
\draw[->](2) -- (4);
\draw[->](2) -- (5);
\draw[->](3) -- (6);
\draw[->](3) -- (7);
\draw[->](4) -- (10);
\draw[->](4) -- (11);
\draw[->](5) -- (12);
\draw[->](5) -- (13)node[circle,fill=black,pos=.5,scale=.4]{};
\draw[->](6) -- (14)node[circle,fill=black,pos=.5,scale=.4]{};
\draw[->](6) -- (15)node[circle,fill=black,pos=.5,scale=.4]{};
\draw[->](7) -- (8);
\draw[->](7) -- (9);
\draw[->](8) -- (16);
\draw[->](8) -- (17)node[circle,fill=black,pos=.5,scale=.4]{};
\draw[->](9) -- (18)node[circle,fill=black,pos=.5,scale=.4]{};
\draw[->](9) -- (19);

\end{tikzpicture}

%% file: psi_macro.tex
\begin{tikzpicture}[>=stealth]

\node[draw, shape=circle](1) at (2.75,4){$\vee$};
\node[draw, shape=circle](2) at (1.25,3){$\vee$};
\node[draw, shape=circle](3) at (4.25,3){$\vee$};
\node[draw, shape=circle](4) at (.5,2){$\wedge$};
\node[draw, shape=circle](5) at (2,2){$\wedge$};
\node[draw, shape=circle](6) at (3.5,2){$\wedge$};
\node[draw, shape=circle](7) at (5,2){$\wedge$};
\node[draw, shape=circle](8) at (4.5,1){$\vee$};
\node[draw, shape=circle](9) at (5.5,1){$\vee$};

\node[draw, shape=circle,white,scale=.5](10) at (0.25,1){\textcolor{black}{\Huge $y_1$}};
\node[draw, shape=circle,white,scale=.5](11) at (.75,1){\textcolor{black}{\Huge $x_1$}};
\node[draw, shape=circle,white,scale=.5](12) at (1.75,1){\textcolor{black}{\Huge $x_1$}};
\node[draw, shape=circle,white,scale=.5](13) at (2.25,1){\textcolor{black}{\Huge $x_2$}};
\node[draw, shape=circle,white,scale=.5](14) at (3.25,1){\textcolor{black}{\Huge $y_1$}};
\node[draw, shape=circle,white,scale=.5](15) at (3.75,1){\textcolor{black}{\Huge $x_1$}};
\node[draw, shape=circle,white,scale=.5](16) at (4.25,0){\textcolor{black}{\Huge $x_2$}};
\node[draw, shape=circle,white,scale=.5](17) at (4.75,0){\textcolor{black}{\Huge $y_2$}};
\node[draw, shape=circle,white,scale=.5](18) at (5.25,0){\textcolor{black}{\Huge $x_2$}};
\node[draw, shape=circle,white,scale=.5](19) at (5.75,0){\textcolor{black}{\Huge $y_2$}};

\draw [rounded corners=4mm,red] (2.75,4.6)--(4.75,3.4)--(4.75,2.6)--(0.75,2.6)--(0.75,3.4) -- cycle;

\draw[red,->](2.75,5) -- (1)node[pos=0.5,right,black]{ $\forall x_1 \forall x_2 \exists y_1 (x_1) \exists y_2 (x_2)$};
\draw[->](1) -- (2);
\draw[->](1) -- (3);
\draw[red,->](2) -- (4);
\draw[red,->](2) -- (5);
\draw[red,->](3) -- (6);
\draw[red,->](3) -- (7);
\draw[->](4) -- (10);
\draw[->](4) -- (11);
\draw[->](5) -- (12);
\draw[->](5) -- (13)node[circle,fill=black,pos=.5,scale=.4]{};
\draw[->](6) -- (14)node[circle,fill=black,pos=.5,scale=.4]{};
\draw[->](6) -- (15)node[circle,fill=black,pos=.5,scale=.4]{};
\draw[->](7) -- (8);
\draw[->](7) -- (9);
\draw[->](8) -- (16);
\draw[->](8) -- (17)node[circle,fill=black,pos=.5,scale=.4]{};
\draw[->](9) -- (18)node[circle,fill=black,pos=.5,scale=.4]{};
\draw[->](9) -- (19);

\end{tikzpicture}

%% file: psi_pushEx.tex
\begin{tikzpicture}[>=stealth]

\node[draw, shape=circle](1) at (2.75,4){$\vee$};
\node[draw, shape=circle](2) at (1.25,3){$\vee$};
\node[draw, shape=circle](3) at (4.25,3){$\vee$};
\node[draw, shape=circle](4) at (.5,2){$\wedge$};
\node[draw, shape=circle](5) at (2,2){$\wedge$};
\node[draw, shape=circle](6) at (3.5,2){$\wedge$};
\node[draw, shape=circle](7) at (5,2){$\wedge$};
\node[draw, shape=circle](8) at (4.5,1){$\vee$};
\node[draw, shape=circle](9) at (5.5,1){$\vee$};

\node[draw, shape=circle,white,scale=.5](10) at (0.25,1){\textcolor{black}{\Huge $y_1$}};
\node[draw, shape=circle,white,scale=.5](11) at (.75,1){\textcolor{black}{\Huge $x_1$}};
\node[draw, shape=circle,white,scale=.5](12) at (1.75,1){\textcolor{black}{\Huge $x_1$}};
\node[draw, shape=circle,white,scale=.5](13) at (2.25,1){\textcolor{black}{\Huge $x_2$}};
\node[draw, shape=circle,white,scale=.5](14) at (3.25,1){\textcolor{black}{\Huge $y_1'$}};
\node[draw, shape=circle,white,scale=.5](15) at (3.75,1){\textcolor{black}{\Huge $x_1$}};
\node[draw, shape=circle,white,scale=.5](16) at (4.25,0){\textcolor{black}{\Huge $x_2$}};
\node[draw, shape=circle,white,scale=.5](17) at (4.75,0){\textcolor{black}{\Huge $y_2$}};
\node[draw, shape=circle,white,scale=.5](18) at (5.25,0){\textcolor{black}{\Huge $x_2$}};
\node[draw, shape=circle,white,scale=.5](19) at (5.75,0){\textcolor{black}{\Huge $y_2$}};

\draw [rounded corners=4mm,red] (2.75,4.6)--(4.75,3.4)--(4.75,2.6)--(0.75,2.6)--(0.75,3.4) -- cycle;

\draw[red,->](2.75,5) -- (1)node[pos=.5,right,black]{ $\forall x_1 \forall x_2$};
\draw[->](1) -- (2);
\draw[->](1) -- (3);
\draw[red,->](2) -- (4)node[pos=.5,left,black]{ $\exists y_1 (x_1)$};
\draw[red,->](2) -- (5);
\draw[red,->](3) -- (6)node[pos=.5,left,black]{ $\exists y_1' (x_1)$};
\draw[red,->](3) -- (7)node[pos=.5,right,black]{ $\exists y_2 (x_2)$};
\draw[->](4) -- (10);
\draw[->](4) -- (11);
\draw[->](5) -- (12);
\draw[->](5) -- (13)node[circle,fill=black,pos=.5,scale=.4]{};
\draw[->](6) -- (14)node[circle,fill=black,pos=.5,scale=.4]{};
\draw[->](6) -- (15)node[circle,fill=black,pos=.5,scale=.4]{};
\draw[->](7) -- (8);
\draw[->](7) -- (9);
\draw[->](8) -- (16);
\draw[->](8) -- (17)node[circle,fill=black,pos=.5,scale=.4]{};
\draw[->](9) -- (18)node[circle,fill=black,pos=.5,scale=.4]{};
\draw[->](9) -- (19);

\end{tikzpicture}

%% file: psi_pushAll.tex
\begin{tikzpicture}[>=stealth]
\node[draw, shape=circle](1) at (2.75,4){$\vee$};
\node[draw, shape=circle](2) at (1.25,3){$\vee$};
\node[draw, shape=circle](3) at (4.25,3){$\vee$};
\node[draw, shape=circle](4) at (.5,2){$\wedge$};
\node[draw, shape=circle](5) at (2,2){$\wedge$};
\node[draw, shape=circle](6) at (3.5,2){$\wedge$};
\node[draw, shape=circle](7) at (5,2){$\wedge$};
\node[draw, shape=circle](8) at (4.5,1){$\vee$};
\node[draw, shape=circle](9) at (5.5,1){$\vee$};

\node[draw, shape=circle,white,scale=.5](10) at (0.25,1){\textcolor{black}{\Huge $y_1$}};
\node[draw, shape=circle,white,scale=.5](11) at (.75,1){\textcolor{black}{\Huge $x_1$}};
\node[draw, shape=circle,white,scale=.5](12) at (1.75,1){\textcolor{black}{\Huge $y_1'$}};
\node[draw, shape=circle,white,scale=.5](13) at (2.25,1){\textcolor{black}{\Huge $x_1$}};
\node[draw, shape=circle,white,scale=.5](14) at (3.25,1){\textcolor{black}{\Huge $x_1$}};
\node[draw, shape=circle,white,scale=.5](15) at (3.75,1){\textcolor{black}{\Huge $x_2$}};
\node[draw, shape=circle,white,scale=.5](16) at (4.25,0){\textcolor{black}{\Huge $x_2$}};
\node[draw, shape=circle,white,scale=.5](17) at (4.75,0){\textcolor{black}{\Huge $y_2$}};
\node[draw, shape=circle,white,scale=.5](18) at (5.25,0){\textcolor{black}{\Huge $x_2$}};
\node[draw, shape=circle,white,scale=.5](19) at (5.75,0){\textcolor{black}{\Huge $y_2$}};

\draw[red,->](2.75,5) -- (1)node[pos=.5,right,black]{ $\forall x_1$};
\draw[->](1) -- (2);
\draw[red,->](1) -- (3)node[pos=.5,right,black]{ $\forall x_2$};
\draw[red,->](2) -- (4)node[pos=.5,left,black]{ $\exists y_1 (x_1)$};
\draw[green,->](2) -- (5)node[pos=.5,right,black]{ $\exists y_1' (x_1)$};
\draw[green,->](3) -- (6);
\draw[->](3) -- (7)node[pos=.5,right,black]{ $\exists y_2 (x_2)$};
\draw[->](4) -- (10);
\draw[->](4) -- (11);
\draw[->](5) -- (12)node[circle,fill=black,pos=.5,scale=.4]{};
\draw[->](5) -- (13)node[circle,fill=black,pos=.5,scale=.4]{};
\draw[->](6) -- (14);
\draw[->](6) -- (15)node[circle,fill=black,pos=.5,scale=.4]{};
\draw[->](7) -- (8);
\draw[->](7) -- (9);
\draw[->](8) -- (16);
\draw[->](8) -- (17)node[circle,fill=black,pos=.5,scale=.4]{};
\draw[->](9) -- (18)node[circle,fill=black,pos=.5,scale=.4]{};
\draw[->](9) -- (19);

\draw [rounded corners=4mm,red] (2.75,4.65)--(3.4,3.75)--(1.25,2.35)--(0.6,3.25)--cycle;

\node(white)at(2.75,5){\textcolor{white}{white}};
\end{tikzpicture}

%% file: psi_final.tex
\begin{tikzpicture}[>=stealth]
\node[draw, shape=circle](1) at (2.75,4){$\vee$};
\node[draw, shape=circle](2) at (1.25,3){$\vee$};
\node[draw, shape=circle](3) at (4.25,3){$\vee$};
\node[draw, shape=circle](4) at (.5,2){$\wedge$};
\node[draw, shape=circle](5) at (2,2){$\wedge$};
\node[draw, shape=circle](6) at (3.5,2){$\wedge$};
\node[draw, shape=circle](7) at (5,2){$\wedge$};
\node[draw, shape=circle](8) at (4.5,1){$\vee$};
\node[draw, shape=circle](9) at (5.5,1){$\vee$};

\node[draw, shape=circle,white,scale=.5](10) at (0.25,1){\textcolor{black}{\Huge $y_1$}};
\node[draw, shape=circle,white,scale=.5](11) at (.75,1){\textcolor{black}{\Huge $x_1$}};
\node[draw, shape=circle,white,scale=.5](12) at (1.75,1){\textcolor{black}{\Huge $y_1'$}};
\node[draw, shape=circle,white,scale=.5](13) at (2.25,1){\textcolor{black}{\Huge $x_1$}};
\node[draw, shape=circle,white,scale=.5](14) at (3.25,1){\textcolor{black}{\Huge $x_1$}};
\node[draw, shape=circle,white,scale=.5](15) at (3.75,1){\textcolor{black}{\Huge $x_2$}};
\node[draw, shape=circle,white,scale=.5](16) at (4.25,0){\textcolor{black}{\Huge $x_2$}};
\node[draw, shape=circle,white,scale=.5](17) at (4.75,0){\textcolor{black}{\Huge $y_2$}};
\node[draw, shape=circle,white,scale=.5](18) at (5.25,0){\textcolor{black}{\Huge $x_2$}};
\node[draw, shape=circle,white,scale=.5](19) at (5.75,0){\textcolor{black}{\Huge $y_2$}};

\draw[->](2.75,5) -- (1)node[pos=.5,right]{ $\forall x_1$};
\draw[->](1) -- (2);
\draw[->](1) -- (3)node[pos=.5,right]{ $\forall x_2$};
\draw[->](2) -- (4);
\draw[->](2) -- (5);
\draw[->](3) -- (6);
\draw[->](3) -- (7)node[pos=.5,right]{ $\exists y_2 (x_2)$};
\draw[->](4) -- (10)node[pos=.5,left]{ $\exists y_1 (x_1)$};
\draw[->](4) -- (11);
\draw[->](5) -- (12)node[circle,fill=black,pos=.5,scale=.4]{}node[pos=.5,left]{$\exists y_1' (x_1)$};
\draw[->](5) -- (13)node[circle,fill=black,pos=.5,scale=.4]{};
\draw[->](6) -- (14);
\draw[->](6) -- (15)node[circle,fill=black,pos=.5,scale=.4]{};
\draw[->](7) -- (8);
\draw[->](7) -- (9);
\draw[->](8) -- (16);
\draw[->](8) -- (17)node[circle,fill=black,pos=.5,scale=.4]{};
\draw[->](9) -- (18)node[circle,fill=black,pos=.5,scale=.4]{};
\draw[->](9) -- (19);

\node(white)at(2.75,5){\textcolor{white}{white}};

\end{tikzpicture}

%% file: 05-experiments.tex
\section{Experimental Results}
\label{sec:experiments}

\noindent We embedded our algorithm into the
DQBF solver HQS, which was the winner of the DQBF track at the QBF\-EVAL'18 and '19
competitions \cite{qbfeval18,qbfeval19}. HQS
includes the powerful DQBF-preprocessor HQSpre~\cite{wimmer-et-al-tacas-2017,wimmer-et-al-jsat-2019}.
After preprocessing has finished, we call the algorithm \textit{DQBFQuantLocalization} to simplify the formula.
HQS augmented with the localization of quantifiers is denoted as \textit{\HQSnp}.\footnotemark
\footnotetext{A recent binary of \HQSnp and all DQBF benchmarks we used are provided at \url{https://abs.informatik.uni-freiburg.de/src/projects_view.php?projectID=21}}

The experiments were run on one core of an Intel Xeon CPU E5-2650v2 with 2.6~GHz.
The runtime per benchmark was limited to 30~min and the memory consumption to 4~GB. We tested our theory with the same 4811 instances as in
\cite{WimmerKBS017}
\cite{gitina-et-al-sat-skolem-2016}
\cite{wimmer-et-al-sat_tr-2015}
\cite{gitina-et-al-date-2015}.
They encompass equivalence checking problems for incomplete circuits
\cite{SchollB01,gitina-et-al-iccd-2013,FinkbeinerT14,FrohlichKBV14},
controller synthesis problems \cite{BloemKS14}
and instances from \cite{BalabanovJ15} where a DQBF has been obtained by succinctly encoding a SAT problem.

Out of 4811 DQBF instances we focus here on those 991 which actually reach our algorithm.
The remaining ones are solved by the preprocessor HQSpre or already exceed the time / memory limit either during preprocessing or while translating the formula into an AIG, \ie
in those cases the results for HQS and \HQSnp do not differ.

When we reach the function \textit{DQBFQuantLocalization} from Alg.~\ref{algo:DQBFQuLo}, for 989 out of 991 instances we can perform the localization of quantifiers.
Quantifier localization enables the elimination of variables in subformulas of 
848 instances.
For 57971 times local quantifier elimination takes place and reduces the 
overall
number of variables in all 
848 benchmarks
which allow variable elimination in subformulas.
49107 variables have been eliminated in this manner.
Note that if a variable has been doubled according to Theorems~\ref{th:forall_and} or \ref{th:equisat:exists_or} and not all of the duplicates can be eliminated locally,
this variable is not counted among the eliminated variables as some duplicates will be pulled back to the root.

\begin{figure}[tb]
  \centering
  \begin{subfigure}{0.45\textwidth}
    \centering
    \includegraphics[scale=0.9]{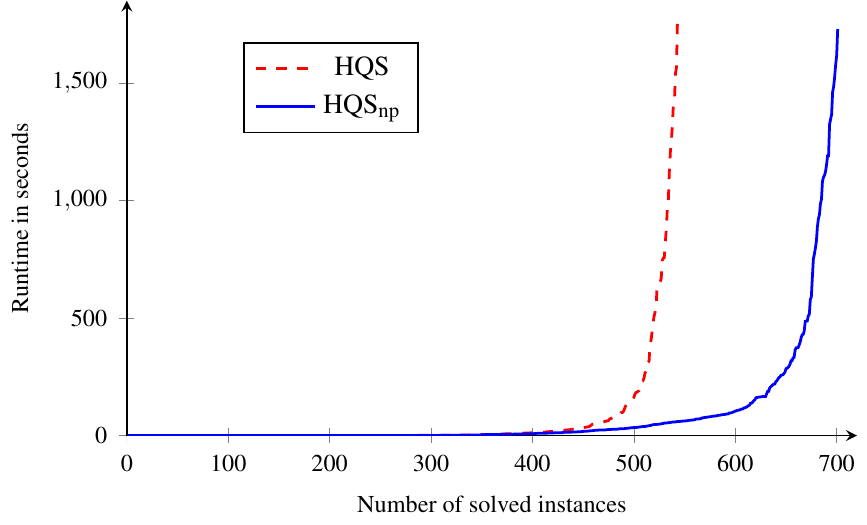}
    \caption{HQS vs.\ \HQSnp\ -- solved instances} \label{fig:hqsvshqsEQ}
  \end{subfigure}
  \hfill
  \begin{subfigure}{0.4\textwidth}
    \centering
    \includegraphics[scale=0.9]{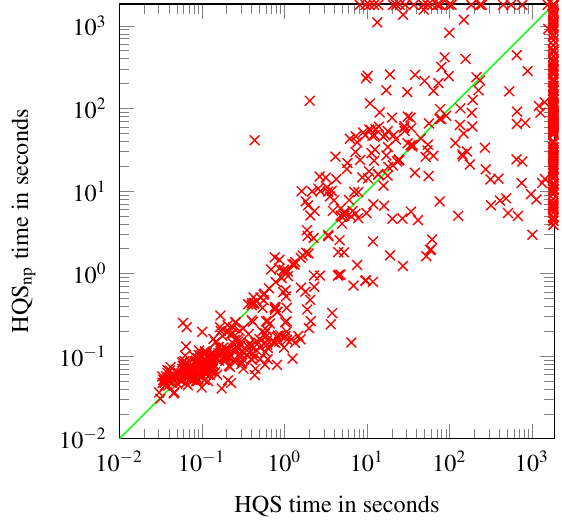}
    \caption{HQS vs.\ \HQSnp\ -- computation time} \label{fig:hqsvshqsEQScatter}
  \end{subfigure}
  \caption{Impact of Quantifier Localization}
\end{figure}

Altogether 701 instances out of 991 were solved by \HQSnp in the end, whereas HQS could only solve 542.
This increases the number of solved instances by more than 29\% (for a cactus plot comparing HQS with \HQSnp see Figure~\ref{fig:hqsvshqsEQ}).
The largest impact of quantifier localization has been observed on equivalence checking benchmarks for incomplete circuits
from \cite{FinkbeinerT14}.

Figure~\ref{fig:hqsvshqsEQScatter} shows the computation times of HQS resp.~\HQSnp for all individual benchmark instances.
The figure reveals that quantifier localization, in its current implementation, does not lead to a better result in every case. 186 instances have been solely solved by \HQSnp, but the opposite is true for 27 benchmarks.
In all of these 27 instances the AIG sizes have grown during local quantifier elimination,
and processing larger AIGs resulted in larger run times.
Altogether, the size of the AIG after \textit{DQBFQuantLocalization} has been decreased in 545 cases and increased in 300 cases (in 3 modified instances the number of AIG nodes did not change),
although in general it is not unusual that the symbolic elimination of quantifiers
must be paid by increasing the sizes of AIGs.
Nevertheless, Figure~\ref{fig:hqsvshqsEQScatter} shows that in most cases the run times of
\HQSnp are faster than those of HQS (and, as already mentioned, the number of solved instances is increased 
by more than 29\%).

We also tested our algorithm on the competition benchmarks from QBFEVAL'18 to QBFEVAL'20~\cite{qbfeval18,qbfeval19,qbfeval20}.
Here the situation is pretty similar.
In those competitions 660 different benchmark instances have been used.
141 out of 660 benchmark instances reach our core algorithm and 65 of them could be solved by the original HQS algorithm.
In all those 141 instances variables are pushed into the formula,
and in 39 instances pushing enabled 3641 local eliminations of variables in total. This made it possible to newly solve 13 benchmark instances and to decrease the runtime for further 10 instances.

%% file: 06-conclusion.tex
\section{Conclusions}
\label{sec:conclusion}
\noindent In this paper, we presented syntax and semantics of non-prenex DQBFs
and proved rules to transform prenex DQBFs into non-prenex DQBFs.
We could demonstrate that we can achieve significant improvements by extending the DQBF solver HQS based on this theory.
Simplifications of DQBFs were due to symbolic quantifier
eliminations that were enabled by pushing quantifiers into the formula based on our rules for non-prenex DQBFs.

In the future, we aim at improving the results of quantifier localization, \eg by introducing
estimates on costs and benefits of quantifier localization operations as well as local quantifier elimination and
by using limits on the growth of AIG sizes caused by local quantifier elimination.

%% file: 06a-acknowledgements.tex
\section*{Acknowledgements}
\label{sec:acknowledgements}
\noindent 
This work was partly supported by the German Research Council (DFG)
as part of the project ``Solving Dependency Quantified Boolean Formulas''
(WI 4490/1-1, SCHO 894/4-1) 
and by the Ministry of Education, Youth and Sports of Czech Republic project ERC.CZ no. LL1908.

%% file: 07-appendix.tex
\appendix
\section{Proof of Theorem~\ref{th:semantics}}
\label{app:semproof}

\semantics*

\begin{proof}
  We show that $\semdef{\psi}=\semth{\psi}$ holds by induction on the structure of $\psi$.
\begin{description}
  \item[\eqref{th:semantics:var}:]
    $v$ is a free variable in $\psi$. Therefore $\sfunc{\psi} = \{v\mapsto \fzero,v\mapsto\fone\}$.
    Only replacing $v$ by $\fone$ turns
    $\psi$ into a tautology, \ie $\semdef{\psi}=\{v\mapsto\fone\} = \semth{\psi}$.

\bigskip
  \item[\eqref{th:semantics:notvar}:]
    Like in the first case, $v$ is a free variable in $\psi$. Therefore $\sfunc{\psi} = \{v\mapsto\fzero,v\mapsto \fone\}$.
    Only replacing $v$ by $\fzero$ turns $\neg v$
    into a tautology, \ie $\semdef{\psi}=\{v\mapsto\fzero\} = \semth{\psi}$.

\bigskip
  \item[\eqref{th:semantics:and} $\psi = (\varphi_1 \land \varphi_2)$:] ~ \\
  $ \semdef{\psi} = \bigl\{ s\in\sfunc{\psi}\,\big|\,\vDash s(\psi)\bigr\} = \bigl\{ s\in\sfunc{\psi}\,\big|\,\vDash s(\varphi_1)\land s(\varphi_2)\bigr\}$. \\
  The conjunction $s(\varphi_1)\land s(\varphi_2)$ is a tautology iff both $s(\varphi_1)$ and $s(\varphi_2)$ are tautologies, \ie
  $\semdef{\psi} = \bigl\{ s\in\sfunc{\psi}\,\big|\,{\vDash s(\varphi_1)} \land {\vDash s(\varphi_2)}\bigr\}$.
  We can restrict $s$ to the variables that actually occur in the sub-formulas, \ie
  $\semdef{\psi} = \bigl\{ s\in\sfunc{\psi}\,\big|\,{\vDash s_{|\varfree[\varphi_1]\dcup\varex[\varphi_1]}(\varphi_1)} \land
            {\vDash s_{|\varfree[\varphi_2]\dcup\varex[\varphi_2]}(\varphi_2)} \bigr\}$. 
  By using Definition~\ref{def:sem} of $\semdef{\cdot}$:
  $\semdef{\psi} = \bigl\{s\in\sfunc{\psi}\,\big|\,s_{|\varfree[\varphi_1]\dcup\varex[\varphi_1]}\in\semdef{\varphi_1} \land
            s_{|\varfree[\varphi_2]\dcup\varex[\varphi_2]}\in\semdef{\varphi_2}\bigr\}$.
  Due to the induction assumption we have $\semdef{\varphi_1}=\semth{\varphi_1}$ and $\semdef{\varphi_2}=\semth{\varphi_2}$ and thus:
  $\semdef{\psi} = \bigl\{s\in\sfunc{\psi}\,\big|\,s_{|\varfree[\varphi_1]\dcup\varex[\varphi_1]}\in\semth{\varphi_1} \land
            s_{|\varfree[\varphi_2]\dcup\varex[\varphi_2]}\in\semth{\varphi_2}\bigr\}$.
  With the definition of $\semth{\cdot}$ in \eqref{th:semantics:and} we finally obtain: \\
  $\semdef{\psi} = \bigl\{s\in\sfunc{\psi}\,\big|\,s\in\semth{\psi}\bigr\} = \semth{\psi}.$

\bigskip
  \item[\eqref{th:semantics:or} $\psi = (\varphi_1 \lor \varphi_2)$:] ~\\
  This case is analogous to the previous case, however it needs an additional argument.
  Here we need the statement `The disjunction $s(\varphi_1)\lor s(\varphi_2)$ is a tautology
  iff $s(\varphi_1)$ or $s(\varphi_2)$ are tautologies' which is not true in general.
  Nevertheless, we can prove it here with the following argument:
  $s(\varphi_1)$ only contains variables from $\varall[\varphi_1]$,
  and similarly $s(\varphi_2)$ only variables from $\varall[\varphi_2]$.
  According to our assumption from Definition~\ref{def:syntax}
  $\varall[\varphi_1]\cap\varall[\varphi_2]=\emptyset$ holds.
  Therefore $s(\varphi_1)\lor s(\varphi_2)$ is a tautology iff at least one of its parts is a tautology.

\bigskip
  \item[\eqref{th:semantics:exists} $\psi =  \exists v(D_v):\varphi^{-v}$:]~\\
    $\semdef{\psi} = \bigl\{s\in\sfunc{\psi}\,\big|\,\vDash s(\exists v(D_v):\varphi^{-v})\bigr\}$.
    The first observation is that $\sfunc{\psi} = \sfunc{\varphi^{-v}}$, since
    $D_v \subseteq V\setminus(\varex[\varphi]\dcup\varall[\varphi]\dcup\{v\})$, \ie $D_v \cap \varall[\varphi] = \emptyset$, and thus
    Skolem function candidates for $v$ are restricted to constant functions,
    no matter whether $v$ is a free variable as in $\varphi^{-v}$ or an existential variable without universal variables in its dependency set as in $\psi$.
    For all other existential variables in $\varphi$, $\varphi^{-v}$ removes $v$ from the dependency sets of all existential variables, but this does not have any effect on the corresponding
    Skolem function candidates, since $v \notin \varall[\varphi]$.
    Second, for each $s\in\sfunc{\psi}$ we have $s\bigl(\exists v(D_v):\varphi\bigr)=s\bigl(\varphi^{-v}\bigr)$. Therefore we get:
    $\semdef{\psi} = \bigl\{s\in\sfunc{\varphi^{-v}}\,\big|\,\vDash s(\varphi^{-v})\bigr\} = \semdef{\varphi^{-v}}$.
    By applying the induction assumption we get
    $\semdef{\varphi^{-v}} = \semth{\varphi^{-v}}$
    and finally, because of the definition of $\semth{\cdot}$:
    $\semdef{\psi}  = \semth{\psi}$.

\bigskip
  \item[\eqref{th:semantics:forall} $\psi = \forall v:\varphi$:]~\\
    $\semdef{\psi} = \bigl\{t\in\sfunc{\psi}\,\big|\,\vDash t(\forall v:\varphi)\bigr\}
                   = \bigl\{t\in\sfunc{\psi}\,\big|\,\vDash t(\varphi)\bigr\}
                   = \bigl\{t\in\sfunc{\psi}\,\big|\,\vDash t(\varphi)[\sfrac{0}{v}]
                     \ \land \vDash t(\varphi)[\sfrac{1}{v}]\bigr\}$.
    For a function $t\in\sfunc{\psi}$, we define two functions $s_0^t,s_1^t\in\sfunc{\varphi}$ by:
        $s_0^t(v)=\fzero$, $s_1^t(v)=\fone$,
        $s_0^t(w)=s_1^t(w)=t(w)$ if $w\in\varex[\varphi]$ with $v\not\in D_v$ or $w\in\varfree[\varphi]\setminus\{v\}$, and
        $s_0^t(w)= t(w)[\sfrac{0}{v}]$, $s_1^t(w)=t(w)[\sfrac{1}{v}]$ for $w\in\varex[\varphi]$ with $v\in D_w$. Then we have:
        $t(\varphi)[\sfrac{0}{v}] = s_0^t(\varphi)$ and $t(\varphi)[\sfrac{1}{v}] = s_1^t(\varphi)$.
        $\semdef{\psi} = \bigl\{t\in\sfunc{\psi}\,\big|\,{\vDash s_0^t(\varphi)} \land {\vDash s_1^t(\varphi)}\bigr\}
                  = \bigl\{t\in\sfunc{\psi}\,\big|\,s_0^t\in\semdef{\varphi} \land s_1^t\in\semdef{\varphi}\bigr\}$.
    The induction assumption gives us: $\semdef{\varphi}=\semth{\varphi}$ and therefore:
     $\semdef{\psi} = \bigl\{t\in\sfunc{\psi}\,\big|\,s_0^t\in\semth{\varphi} \land s_1^t\in\semth{\varphi}\bigr\}$.
    With the equality $t(w) = \mathrm{ITE}\bigl(v,t(w)[\sfrac{1}{v}], t(w)[\sfrac{0}{v}]\bigr) = \mathrm{ITE}\bigl(v,s_1^t(w), s_0^t(w)\bigr)$ for $w\in\varex[\varphi]=\varex[\psi]$ with $v\in D_w$
              and $t(w)=s_0^t(w)=s_1^t(w)$ for the remaining existential or free variables, we obtain:
     \begin{align*}
     \semdef{\psi} &= \bigl\{t\in\sfunc{\psi}\,\big|\,\exists s_0,s_1\in\semth{\varphi}: s_0(v)=\fzero\land s_1(v)=\fone \\
         &\quad \land \; t(w) = s_0(w)=s_1(w)\text{ for } w\in\varfree[\psi] = \varfree[\varphi]\setminus\{v\} \\
         &\quad \land \; t(w) = s_0(w)=s_1(w)\text{ for } w\in\varex[\psi]\text{ with }v\notin D_w \\
         &\quad \land \; t(w) = \mathrm{ITE}\bigl(v,s_1(w), s_0(w)\bigr) \text{ for } w\in\varex[\psi]\text{ with }v\in D_w\bigr\} \\
         &= \semth{\psi}.
     \end{align*}
  \end{description}
\end{proof}

\section{Proof of Theorem~\ref{th:rules}}
\label{app:ruleproof}

\rules*

\begin{proof}
  \begin{description}
  \item[\eqref{equiv:indep_exists}:]
  Since the left-hand side $\exists y(D_y): \varphi$ of rule \eqref{equiv:indep_exists} is a well-defined DQBF, we
  have $\varphi = \varphi^{-y}$ according to Definition~\ref{def:syntax}.
  The rule simply follows from $\sem{\exists y(D_y):\varphi^{-y}} = \sem{\varphi^{-y}}$ shown in
  Theorem~\ref{th:semantics}.

  \item[\eqref{equiv:indep}:]
  We omit the proof here, as it immediately follows from the more general Theorem~\ref{th:equiv:indep}.

  \item[\eqref{th:rules3}:]
  The statement easily follows from Theorem~\ref{th:semantics}, \eqref{th:semantics:forall} considering
  that $\varphi$ contains only free variables.
  Let $\psi_1 \colonequals \forall x : \varphi$ and
  $\psi_2 \colonequals \varphi[\sfrac{0}{x}] \land \varphi[\sfrac{1}{x}]$.
  We have $\sfunc{\psi_1} = \sfunc{\psi_2}$ and
  $\sem{\psi_1}
  = \bigl\{ t\in\sfunc{\psi_1}\,\big|\,
          \exists s_0, s_1\in\sem{\varphi}:  s_0(x)=\fzero\land s_1(x)=\fone \; \land
          \forall w\in\varfree[\psi_1]: t(w)\colonequals s_0(w) = s_1(w) \bigr\}
  = \sem{\psi_2}$.

  \item[\eqref{equiv:exists}:]
  Let $\psi_1 \colonequals \exists y(D_y) : \varphi$ and
  $\psi_2 \colonequals \varphi[\sfrac{0}{y}] \lor \varphi[\sfrac{1}{y}]$.

  Assume that $\sem{\psi_1}\neq\emptyset$ and let $t\in\sem{\psi_1}$.
  $t(y)$ has to be a constant function, \ie $t(y) = \fzero$ or $t(y) = \fone$.
  We choose a Skolem function $t'$ for $\psi_2$ by
  $t'(v) = t(v)$ for all $v\in \varfree[\psi_2]$.
  It is clear that $t'$ is a Skolem function candidate for $\psi_2$.
  Assume \wlogen that $t(y) = \fzero$.
  Since $t\bigl(\exists y(D_y) : \varphi\bigr) = t\bigl(\varphi\bigr)$ is a tautology,
  $t'\bigl(\varphi[\sfrac{0}{y}]\bigr) = t\bigl(\varphi\bigr)$ is a tautology, too.
  Thus, $t'\bigl(\varphi[\sfrac{0}{y}] \lor \varphi[\sfrac{1}{y}]\bigr)
     = t'\bigl(\varphi[\sfrac{0}{y}]\bigr) \lor t'\bigl(\varphi[\sfrac{1}{y}]\bigr)$
  is a tautology. This shows that $t'\in\sem{\psi_2}$ and therefore $\sem{\psi_2}\neq\emptyset$.

  For the opposite direction assume that $\sem{\psi_2}\neq\emptyset$ and let $t'\in\sem{\psi_2}$.
  $t'\bigl(\varphi[\sfrac{0}{y}] \lor \varphi[\sfrac{1}{y}]\bigr)$ is a tautology.
  Since $\varphi$ contains only free variables, $t'\bigl(\varphi[\sfrac{0}{y}]\bigr)$ and
  $t'\bigl(\varphi[\sfrac{1}{y}]\bigr)$ are (equivalent to) constants. At least one of them is $\fone$,
  assume \wlogen $t'\bigl(\varphi[\sfrac{0}{y}]\bigr)$. Now we choose
  $t(v) = t'(v)$ for all $v \in \varfree[\psi_2] \setminus \{y\}$ and $t(y) = \fzero$.
  Since $t\bigl(\exists y(D_y) : \varphi\bigr) = t\bigl(\varphi\bigr) = t'\bigl(\varphi[\sfrac{0}{y}]\bigr) = \fone$,
  we have $t\in\sem{\psi_1}$ and therefore $\sem{\psi_1}\neq\emptyset$.

 \item[\eqref{equiv:forall_and}:]
 We omit the proof here, since it immediately follows from the more general Theorem~\ref{th:forall_and}.

 \item[\eqref{equiv:forall_and2}:]
 We omit the proof here, since it immediately follows from the more general Theorem~\ref{th:equisat:forall_and}.

 \item[\eqref{equiv:forall_or}:]
   Let $\psi_1\colonequals\forall x:(\varphi_1\op \varphi_2)$ and $\psi_2\colonequals\bigl(\varphi_1\op (\forall x:\varphi_2)\bigr)$
   and assume that $x\notin  V_{\varphi_1}$ and $x\notin D_y$ for any $y\in\varex[\varphi_1]$.
   Note that we need $x\notin  V_{\varphi_1}$, since otherwise $\psi_2$ would not be
   well-formed according to Definition~\ref{def:syntax}.
   From $x\notin D_y$ for any $y\in\varex[\varphi_1]$ we conclude that $\sfunc{\psi_1}=\sfunc{\psi_2}$.
   Then we have:
   $\sem{\psi_1}
    = \bigl\{s\in\sfunc{\psi_1}\,\big|\,\vDash s(\forall x:(\varphi_1\op \varphi_2))\bigr\}
    = \bigl\{s\in\sfunc{\psi_1}\,\big|\,\vDash s(\varphi_1\op \varphi_2)\bigr\}
    = \bigl\{s\in\sfunc{\psi_1}\,\big|\,\vDash s(\varphi_1 \op (\forall x: \varphi_2))\bigr\}
    = \bigl\{s\in\sfunc{\psi_2}\,\big|\,\vDash s(\varphi_1 \op (\forall x: \varphi_2))\bigr\}$,
    since $\sfunc{\psi_1}=\sfunc{\psi_2}$, and finally $\sem{\psi_1} = \sem{\psi_2}$.

 \item[\eqref{equiv:exists_or}:]
 We set $\psi_1\colonequals \exists y(D_y):(\varphi_1\lor\varphi_2)$
 and  $\psi_2\colonequals \bigl(\exists y(D_y):\varphi_1\bigr)\lor\bigl(\exists y'(D_y):\varphi_2[\sfrac{y'}{y}]\bigr)$.
 \begin{align*}
 \sem{\psi_1} &= \bigl\{t\in\sfunc{\psi_1}\,\big|\,{\vDash t(\exists y(D_y):(\varphi_1\lor\varphi_2))} \bigr\} \\
 &= \bigl\{t\in\sfunc{\psi_1}\,\big|\,{\vDash t(\varphi_1\lor\varphi_2)}\bigr\} \\
 &= \bigl\{t\in\sfunc{\psi_1}\,\big|\,{\vDash t(\varphi_1)} \lor {\vDash t(\varphi_2)}\bigr\}
 \intertext{The last equality holds, because the variables occurring in $t(\varphi_1)$ and $t(\varphi_2)$ are disjoint.
 On the other hand we have}
 \sem{\psi_2} &= \bigl\{t'\in\sfunc{\psi_2}\,\big|\,{\vDash t'((\exists y(D_y):\varphi_1)\lor(\exists y'(D_y):\varphi_2[\sfrac{y'}{y}]))} \bigr\} \\
 &= \bigl\{t'\in\sfunc{\psi_2}\,\big|\,{\vDash t'(\varphi_1 \lor \varphi_2[\sfrac{y'}{y}])}\bigr\} \\
 &= \bigl\{t'\in\sfunc{\psi_2}\,\big|\,{\vDash t'(\varphi_1)} \lor {\vDash t'(\varphi_2[\sfrac{y'}{y}])}\bigr\}.
 \end{align*}
 Again, the last equality holds, because the variables occurring in $t'(\varphi_1)$ and $t'\bigl(\varphi_2[\sfrac{y'}{y}]\bigr)$ are disjoint.

 Assume that $\sem{\psi_1}\neq\emptyset$ and let $t\in\sem{\psi_1}$.
 Then $t(\varphi_1)$ or $t(\varphi_2)$ is a tautology.
 We choose a Skolem function $t'$ for $\psi_2$ by
 $t'(v) = t(v)$ for all
 $v\in\varfree[\psi_1]\dcup\varex[\psi_1]$
 and
 $t'(y') = t(y)$. (Note that $t(y)$ as well as $t'(y')$ have to be constant functions according
 to Def.~\ref{def:skolem_function_candidates}.)
 If $t(\varphi_1)$ is a tautology, then $t'(\varphi_1) = t(\varphi_1)$ is a
 tautology as well, $t'\in\sem{\psi_2}$ and therefore $\sem{\psi_2}\neq\emptyset$.
 If $t(\varphi_1)$ is not a tautology, then $t(\varphi_2)$ has to be a tautology and
 $t'\bigl(\varphi_2[\sfrac{y'}{y}]\bigr) = t(\varphi_2)$ is a tautology as well,
 $t'\in\sem{\psi_2}$ and therefore $\sem{\psi_2}\neq\emptyset$.

 For the opposite direction assume $\sem{\psi_2}\neq\emptyset$ and let $t'\in\sem{\psi_2}$.
 Then $t'(\varphi_1)$ or $t'(\varphi_2[\sfrac{y'}{y}])$ is a tautology.
 If $t'(\varphi_1)$ is a tautology, we choose
 the Skolem function for $\psi_1$ by
 $t(v) = t'(v)$ for all $v\in\varfree[\psi_1]\dcup\varex[\psi_1]$.
 It immediately follows that $t(\varphi_1) = t'(\varphi_1)$ is a tautology as well,
 $t\in\sem{\psi_1}$ and therefore $\sem{\psi_1}\neq\emptyset$.
 If $t'(\varphi_1)$ is not a tautology, then $t'(\varphi_2[\sfrac{y'}{y}])$ has to be a tautology
 and we choose  the Skolem function for $\psi_1$ by
 $t(v) = t'(v)$ for all $v\in (\varfree[\psi_1]\dcup\varex[\psi_1]) \setminus \{y\}$
 and $t(y) = t'(y')$.
 Then $t(\varphi_2) = t'\bigl(\varphi_2[\sfrac{y'}{y}]\bigr)$ is also a tautology
 and again $t\in\sem{\psi_1}$ and therefore $\sem{\psi_1}\neq\emptyset$.

 \item[\eqref{equiv:exists_and}:]
   Let $\psi_1 \colonequals\exists y (D_y):(\varphi_1 \op \varphi_2)$ and
   $\psi_2\colonequals\varphi_1\op \bigl(\exists y(D_y):\varphi_2\bigr)$.
   Note that we need $y\not\in\var[\varphi_1]$, since otherwise $\psi_2$ would not be
   well-formed according to Definition~\ref{def:syntax}.
   The following equalities hold:
   \begin{align*}
      \sem{\psi_1} &= \bigl\{s\in\sfunc{\psi_1}\,\big|\,{\vDash s(\exists y(D_y):(\varphi_1\op\varphi_2))}\bigr\} \\
      &= \bigl\{s\in\sfunc{\psi_1}\,\big|\,{\vDash s(\varphi_1\op\varphi_2)}\bigr\} \\
      &= \bigl\{s\in\sfunc{\psi_1}\,\big|\,\vDash s(\varphi_1\op (\exists y(D_y):\varphi_2)) \bigr\} \\
      &= \bigl\{s\in\sfunc{\psi_2}\,\big|\,\vDash s(\varphi_1\op (\exists y(D_y):\varphi_2)) \bigr\}
         \text{ since $\sfunc{\psi_1}=\sfunc{\psi_2}$} \\
      &= \sem{\psi_2}\,.
   \end{align*}

 \item[\eqref{equiv:exists_exists}:]
   By applying Theorem~\ref{th:semantics}, Equation~\eqref{th:semantics:exists} multiple times, we get:
   $\sem{\exists y_1(D_{y_1})\exists y_2(D_{y_2}):\varphi}
            = \sem{\exists y_2(D_{y_2}):\varphi}
            = \sem{\varphi}
            = \sem{\exists y_1(D_{y_1}):\varphi}
            = \sem{\exists y_2(D_{y_2})\exists y_1(D_{y_1}):\varphi}$.

 \item[\eqref{equiv:forall_forall}:]
   We set
   $\psi_1\colonequals\forall x_1 \forall x_2:\varphi$ and $\psi_2\colonequals\forall x_2\forall x_1:\varphi$. Then we have:
   $\sem{\psi_1} = \sem{\forall x_1 \forall x_2:\varphi}
      = \bigl\{ s\in\sfunc{\psi_1}\,\big|\, {\vDash s(\forall x_1\forall x_2:\varphi)} \bigr\}
      = \bigl\{ s\in\sfunc{\psi_1}\,\big|\, {\vDash s(\varphi)} \bigr\}
      = \bigl\{ s\in\sfunc{\psi_2}\,\big|\, {\vDash s(\varphi)} \bigr\}$,
      since $\sfunc{\psi_1}=\sfunc{\psi_2}$, and then
   $\sem{\psi_1} = \bigl\{ s\in\sfunc{\psi_2}\,\big|\, {\vDash s(\forall x_2\forall x_1:\varphi)} \bigr\} = \sem{\psi_2}$.

 \item[\eqref{equiv:forall_exists}:]
   We set $\psi_1\colonequals \forall x\exists y(D_y):\varphi$ and $\psi_2\colonequals \exists y(D_y)\forall x:\varphi$.

   First note that $\exists y(D_y)\forall x:\varphi$ is not well-formed according to
   Definition~\ref{def:syntax} if $x\in D_y$, because $x$ is universal in $\forall x:\varphi$.
   With $x\notin D_y$ we show that $\sem{\psi_1}=\sem{\psi_2}$.
   We have:
   $\sem{\psi_1} = \bigl\{ s\in\sfunc{\psi_1}\,\big|\,{\vDash s(\forall x\exists y(D_y):\varphi)} \bigr\}
   = \bigl\{ s\in\sfunc{\psi_1}\,\big|\,{\vDash s(\varphi)} \bigr\}$.
   Because $x\notin D_y$, the Skolem function candidates for $y$ in $\psi_1$ are restricted to constant functions.
   The same holds for $y$ in $\psi_2$. Therefore $\sfunc{\psi_1}=\sfunc{\psi_2}$ is true. So we can write:
   $\sem{\psi_1} = \bigl\{ s\in\sfunc{\psi_2}\,\big|\,{\vDash s(\varphi)} \bigr\}
   = \bigl\{ s\in\sfunc{\psi_2}\,\big|\,{\vDash s(\exists y(D_y)\forall x:\varphi)}\bigr\}
   = \sem{\psi_2}$.
\end{description}
\end{proof}

\section{Proof of Theorem~\ref{th:rules2}}
\label{app:rules2}

\rulesTwo*

\begin{proof}
We show equisatisfiability by proving that $\sem{\psi'} \neq \emptyset$ implies
$\sem{\psi} \neq \emptyset$ and vice versa.
First assume that there is a Skolem function $s'\in\sem{\psi'}$ with $\vDash s'(\psi')$.
We define $s \in \sfunc{\psi}$ by $s(v) \colonequals s'(v)$ for all $v \in \varex[\psi'] \cup \varfree[\psi'] \setminus \{y\}$
and $s(y) \colonequals s'(\varphi[\sfrac{1}{y}])$.
Since $\varphi$ contains only variables from $D_y \cup \varfree[\psi] \cup \{v \in \varex[\psi] \; | \; D_v \subseteq D_y\}$,
$\support\bigl(s(y)\bigr) = \support\bigl(s'(\varphi[\sfrac{1}{y}])\bigr) \subseteq D_y$, \ie $s \in \sfunc{\psi}$.
By definition of $s(y)$, $s(\psi)$ is the same as $s'(\psi'')$ where
$\psi''$ results from $\psi$ by replacing the subformula
$\varphi$ by
$\varphi[\sfrac{\varphi[\sfrac{1}{y}]}{y}]$. According to \cite{Jiang09}, quantifier elimination can be done by composition as well
and $\varphi[\sfrac{\varphi[\sfrac{1}{y}]}{y}]$ is equivalent to $\varphi[\sfrac{0}{y}] \lor \varphi[\sfrac{1}{y}]$,
\ie $s(\psi) = s'(\psi')$ and thus $\vDash s(\psi)$.

Now assume $s \in  \sem{\psi}$ with $\vDash s(\psi)$.
Consider $s'$ which results from $s$ by removing $y$ from the domain of $s$.
Then $s'(\varphi)$ can be regarded as a Boolean function depending on $D_y \cup \{y\}$.
$s(\varphi) = s'(\varphi)[\sfrac{s(y)}{y}]$ is a function which
(1) does not depend on $y$ and which (2) has the property that
for each assignment $\mu$ to the variables from $D_y \cup \{y\}$
$\mu\bigl(s'(\varphi)[\sfrac{s(y)}{y}]\bigr) = \mu\bigl(s'(\varphi)\bigr)$
or
$\mu'\bigl(s'(\varphi)[\sfrac{s(y)}{y}]\bigr) = \mu'\bigl(s'(\varphi)\bigr)$
with $\mu'$ resulting from $\mu$ by flipping the assignment to $y$.
$s'(\varphi)[\sfrac{s'(\varphi)[\sfrac{1}{y}]}{y}]$ which corresponds
to the existential quantification of $y$ in $s'(\varphi)$
is \emph{the largest function} fulfilling (1) and (2), \ie
$s(\varphi) \leq s'(\varphi)[\sfrac{s'(\varphi)[\sfrac{1}{y}]}{y}]$.
We derive $s''$ from $s$ by replacing $s(y)$ by $s'(\varphi)[\sfrac{1}{y}]$
and obtain also $s(\psi) \leq s''(\psi)$,
since $\psi$ is in NNF, \ie contains negations only at the inputs, thus
$\varphi$ is not in the scope of any negation in $\psi$, but only in the scope of
conjunctions and disjunctions which are monotonic functions.
Thus $\vDash s(\psi)$ implies $\vDash s''(\psi)$.
Again, due to the equivalence of
$\varphi[\sfrac{\varphi[\sfrac{1}{y}]}{y}]$ and $\varphi[\sfrac{0}{y}] \lor \varphi[\sfrac{1}{y}]$,
we conclude $s''(\psi) = s'(\psi')$ and thus $\vDash s'(\psi')$.
\end{proof}

\section{Proof of Theorem~\ref{th:forall_and}}
\label{app:forall_and}

\forallAnd*

\begin{proof}
First, we assume that $\sem{\psi}\neq\emptyset$ and let $t\in\sem{\psi}$, \ie
$t(\psi)$ is a tautology. For $\psi'$ we construct a Skolem function $t'$ by
$t'(v) \colonequals t(v)[\sfrac{x'}{x}]$ for all $v\in\varex[\varphi_2]$,
$t'(v) \colonequals t(v)$ otherwise. It is easy to see that $t'$ is a Skolem function candidate
for $\psi'$.

Now assume that $t'(\psi')$ is not a tautoloy, \ie there is an assignment
$\mu' \in \assign(\varall[\psi'])$ with $\mu'\bigl(t'(\psi')\bigr) = 0$.
Since $\mu'\bigl(t(\psi)\bigr) = 1$, we have
$\mu'\bigl(t(\psi_1)\bigr) = 1$ and $\mu'\bigl(t'(\psi_2)\bigr) = 0$
according to Lemma~\ref{lemma:monotonic}.
With $\psi_2 = \bigl(\forall x:\varphi_1\bigr)\land\bigl(\forall x':\varphi_2[\sfrac{x'}{x}]\bigr)$
we obtain $\mu'\bigl(t'(\forall x:\varphi_1)\bigr) = \mu'\bigl(t'(\varphi_1)\bigr) = 0$ or
$\mu'\bigl(t'(\forall x':\varphi_2[\sfrac{x'}{x}])\bigr) = \mu'\bigl(t'(\varphi_2[\sfrac{x'}{x}])\bigr) = 0$.
In the first case we
have $\mu'\bigl(t(\varphi_1)\bigr) = \mu'\bigl(t'(\varphi_1)\bigr) = 0$ which contradicts
$\mu'\bigl(t(\psi_1)\bigr) = \mu'\bigl(t(\varphi_1\land\varphi_2)\bigr) = \mu'\bigl(t(\varphi_1)\bigr) \land \mu'\bigl(t(\varphi_2)\bigr) = 1$.
In the second case we define
$\mu \in \assign(\varall[\psi])$ by $\mu(v) = \mu'(v)$ for all $v \in \varall[\psi] \setminus \{x\}$
and $\mu(x) = \mu'(x')$.
In this case we obtain
$\mu\bigl(t(\varphi_2)\bigr) = \mu'\bigl(t'(\varphi_2[\sfrac{x'}{x}]))\bigr) = 0$.
Thus we obtain
$\mu\bigl(t(\forall x:(\varphi_1\land\varphi_2))\bigr) =  \mu\bigl(t(\varphi_1) \land t(\varphi_2)\bigr)
= \mu\bigl(t(\varphi_1)\bigr) \land \mu\bigl(t(\varphi_2)\bigr) = 0$.
Since $\psi$ and $\psi'$ only
differ in the $\psi_1$-/$\psi_2$-part and the only occurrences of $x$ in $t(\psi)$ are in $t(\psi_1)$,
this leads to
$\mu\bigl(t(\psi)\bigr) = \mu'\bigl(t'(\psi')\bigr) = 0$, which is a contradiction to the fact
that $t(\psi)$ is a tautology. Thus, $t'(\psi')$ has to be a tautology as well
and $\sem{\psi'}\neq\emptyset$.

For the opposite direction we assume that
$\sem{\psi'}\neq\emptyset$ and let $t'\in\sem{\psi'}$.
We obtain $\sem{\psi}\neq\emptyset$ with similar arguments:
We construct a Skolem function $t$ for $\psi$ as follows:
$t(v) = t'(v)[\sfrac{x}{x'}]$ for $v\in\varex[{\varphi_2}]$
and $t(v) = t'(v)$ otherwise.
Assume that $t(\psi)$ is not a tautology, \ie there is an assignment
$\mu \in \assign(\varall[\psi])$ with $\mu\bigl(t(\psi)\bigr) = 0$ and
define $\mu' \in \assign(\varall[\psi'])$ by
$\mu'(v) \colonequals \mu(v)$ for all $v \in \varall[\psi]$ as well as $\mu'(x') \colonequals \mu(x)$.
$\mu'\bigl(t'(\psi')\bigr) = 1$ because $t'(\psi')$ is a tautology.
We obtain $\mu'\bigl(t'(\psi_2)\bigr) = 1$ and $\mu'\bigl(t(\psi_1)\bigr) = \mu\bigl(t(\psi_1)\bigr) = 0$
as above by Lemma~\ref{lemma:monotonic}.
From $\mu\bigl(t(\psi_1)\bigr) = 0$ we conclude $\mu\bigl(t(\varphi_1)\bigr) = 0$ or
$\mu\bigl(t(\varphi_2)\bigr) = 0$.
Since $\mu'\bigl(t'(\varphi_1)\bigr) = \mu\bigl(t(\varphi_1)\bigr)$ and
$\mu'\bigl(t'(\varphi_2[\sfrac{x'}{x}])\bigr) = \mu\bigl(t(\varphi_2)\bigr)$,
this implies
$\mu'\bigl(t'((\forall x:\varphi_1)\land(\forall x':\varphi_2[\sfrac{x'}{x}]))\bigr) = 0$
which contradicts $\mu'\bigl(t'(\psi_2)\bigr) = 1$ derived above.
Thus, $t(\psi)$ is a tautology and $\sem{\psi}\neq\emptyset$.
\end{proof}